\documentclass[aps,prx,twocolumn,superscriptaddress]{revtex4-1}
\pdfoutput=1

\usepackage{amsmath,amsfonts,amsthm}
\usepackage{siunitx}
\usepackage{bm}
\usepackage[dvipsnames]{xcolor}
\usepackage{graphicx}
\usepackage{textpos} 

\usepackage[pdfpagelabels,pdftex,bookmarks,breaklinks]{hyperref}
\definecolor{darkblue}{RGB}{0,0,127} 
\definecolor{darkgreen}{RGB}{0,150,0}
\hypersetup{
	colorlinks,
	linkcolor=darkblue,
	citecolor=darkgreen,
	filecolor=red,
	urlcolor=blue,
	pdftitle={Tailoring surface codes for highly biased noise},
	pdfauthor={David K. Tuckett, Andrew S. Darmawan, Christopher T. Chubb, Sergey Bravyi, Stephen D. Bartlett, Steven T. Flammia}
}

\newtheorem{theorem}{Theorem}
\newtheorem{lemma}{Lemma}
\newtheorem{corollary}{Corollary}
\theoremstyle{definition}

\theoremstyle{remark}

\newcommand{\ket}[1]{| #1 \rangle}

\newcommand{\calC}{{\cal C}}
\newcommand{\calG}{{\cal G}}
\newcommand{\calL}{{\cal L}}
\newcommand{\calO}{{\cal O}}
\newcommand{\calR}{{\cal R}}
\newcommand{\ZZ}{{\mathbb{Z}}}
\newcommand{\be}{\begin{equation}}
\newcommand{\ee}{\end{equation}}
\newcommand{\REP}{\mathrm{REP}}
\begin{document}


\title{Tailoring Surface Codes for Highly Biased noise}
\begin{textblock*}{3cm}(-10cm,-1.5cm)
   YITP-19-68
\end{textblock*}



\author{David K. Tuckett}
\affiliation{Centre for Engineered Quantum Systems, School of Physics, The University of Sydney, Sydney, New South Wales 2006, Australia}
\author{Andrew S. Darmawan}
\affiliation{Yukawa Institute for Theoretical Physics (YITP), Kyoto University, Kitashirakawa Oiwakecho, Sakyo-ku, Kyoto 606-8502, Japan}
\affiliation{JST, PRESTO, 4-1-8 Honcho, Kawaguchi, Saitama 332-0012, Japan}
\author{Christopher T. Chubb}
\affiliation{Centre for Engineered Quantum Systems, School of Physics, The University of Sydney, Sydney, New South Wales 2006, Australia}
\author{Sergey Bravyi}
\affiliation{IBM T.J. Watson Research Center, Yorktown Heights, New York 10598, USA}
\author{Stephen D. Bartlett}
\affiliation{Centre for Engineered Quantum Systems, School of Physics, The University of Sydney, Sydney, New South Wales 2006, Australia}
\author{Steven T. Flammia}
\affiliation{Centre for Engineered Quantum Systems, School of Physics, The University of Sydney, Sydney, New South Wales 2006, Australia}
\affiliation{Yale Quantum Institute, Yale University, New Haven, Connecticut 06520, USA}


\date{12 November 2019}

\begin{abstract}
The surface code, with a simple modification, exhibits ultrahigh error-correction thresholds when the noise is biased toward dephasing.
Here, we identify features of the surface code responsible for these ultrahigh thresholds.  We provide strong evidence that the threshold error rate of the surface code tracks the hashing bound exactly for all biases and show how to exploit these features to achieve significant improvement in logical failure rate.
First, we consider the infinite bias limit, meaning pure dephasing.
We prove that the error threshold  of the modified surface code for pure dephasing noise is $50\%$, i.e., that all qubits are fully dephased, and this threshold can be achieved by a polynomial-time decoding algorithm.
We demonstrate that the subthreshold behavior of the code depends critically on the precise shape and boundary conditions of the code.  That is, for rectangular surface codes with standard rough and smooth open boundaries, it is controlled by the parameter $g=\gcd(j,k)$, where $j$ and $k$ are dimensions of the surface code lattice.
We demonstrate a significant improvement in logical failure rate with pure dephasing for \textit{coprime} codes that have $g=1$, and closely-related \textit{rotated} codes, which have a modified boundary.
The effect is dramatic: The same logical failure rate achievable with a square surface code and $n$ physical qubits can be obtained with a coprime or rotated surface code using only $O(\sqrt{n})$ physical qubits.
Finally, we use approximate maximum-likelihood decoding to demonstrate that this improvement persists for a general Pauli noise biased toward dephasing.
In particular, comparing with a square surface code, we observe a significant improvement in logical failure rate against biased noise using a rotated surface code with approximately half the number of physical qubits.
\end{abstract}


\maketitle


\section{\label{sec:introduction}Introduction}
Quantum error-correcting codes are expected to play a fundamental role in enabling quantum computers to operate at a large scale in the presence of noise.
The surface code~\cite{Bravyi1998}, an example of a topological stabilizer code~\cite{Terhal2015}, is one of the most studied and promising candidates, giving excellence performance for error correction while requiring only check operators (stabilizers) acting on a small number of neighboring qubits~\cite{Dennis2002}.

The error-correction threshold of a code family, which denotes the physical error rate below which the logical failure rate can be made arbitrarily small by increasing the code size, is strongly dependent on the noise model.
The most commonly studied noise model is uniform depolarization of all qubits, where independent single-qubit Pauli $X$, $Y$, and $Z$ errors occur at equal rates.
However, in many quantum architectures such as certain superconducting qubits~\cite{Aliferis2009}, quantum dots~\cite{Shulman2012}, and trapped ions~\cite{Nigg2014}, among others, the noise is biased toward dephasing, meaning that $Z$ errors occur much more frequently than other errors.
Recently, it was shown that, with a simple modification, the surface code exhibits ultrahigh thresholds with such $Z$-biased noise~\cite{Tuckett2018}, where bias is defined as the ratio of the probability of a high-rate $Z$ error over the probability of a low-rate $X$ or $Y$ error.

\begin{figure}[ht]
  \includegraphics{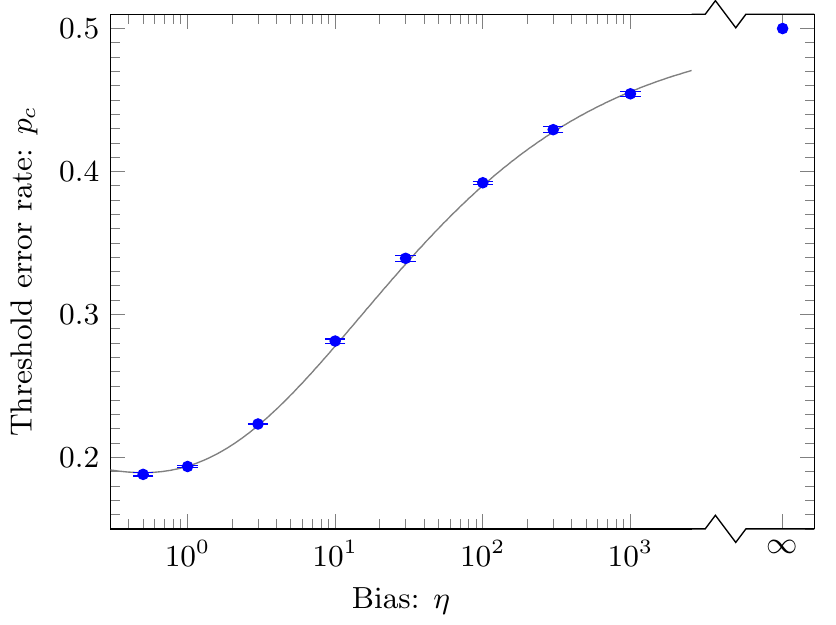}
  \caption{\label{fig:threshold-v-bias}
    Threshold error rate $p_c$ as a function of bias $\eta$.
    Points show threshold estimates for the surface code.
    Error bars indicate one standard deviation relative to the fitting procedure.
    The point at the smallest bias corresponds to $\eta=0.5$ or standard depolarizing noise.
    The point at infinite bias indicates the analytically proven $50\%$ threshold value.
    The gray line is the hashing bound for the associated Pauli error channel.
  }
\end{figure}

\begin{figure*}
  \includegraphics{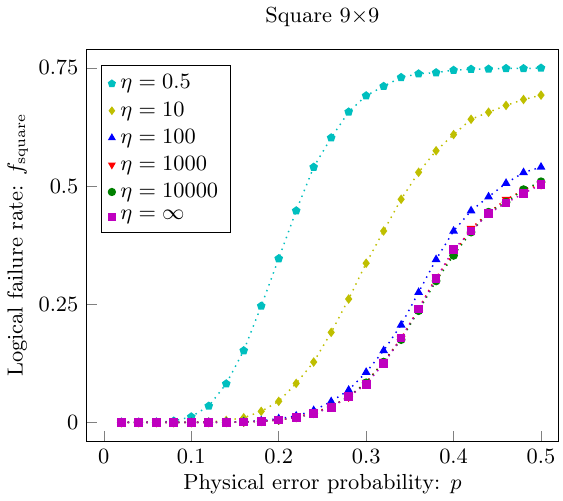}
  \includegraphics{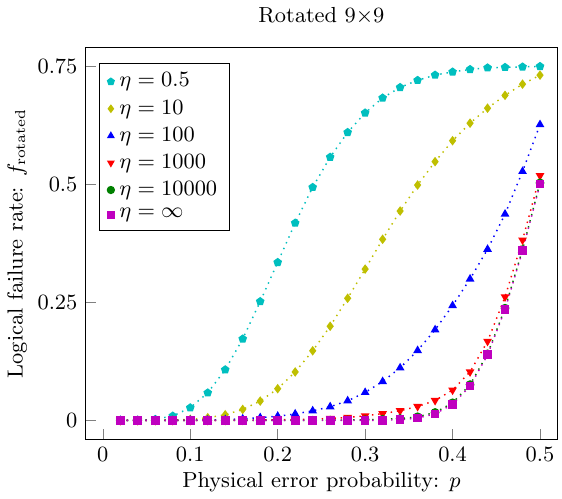}
  \includegraphics{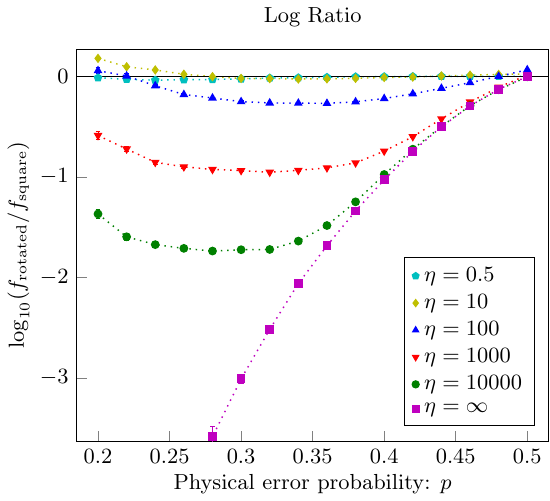}
  \caption{\label{fig:coprime-v-square-small-code-comparison}
    Logical failure rates $f_{\text{square}}$ and $f_{\text{rotated}}$ as a function of physical error probability $p$ for small comparable square and rotated $9{\times}9$ codes and the logarithm of the ratio of logical failure rates $\log_{10}(f_{\text{rotated}} / f_{\text{square}})$ with noise biases $\eta \in \{0.5, 10, 100, 1000, 10\,000, \infty\}$.
    Error bars indicate one standard deviation.
    Data points are sample means over 30\,000 and 1\,200\,000 runs for the square and rotated codes, respectively, using approximate maximum-likelihood decoding converged to within half a standard deviation for both codes.
    Dotted lines connect successive data points for a given $\eta$.
  }
\end{figure*}

In this paper, we identify and characterize the features of the noise-tailored surface code that contribute to its ultrahigh thresholds with $Z$-biased noise and demonstrate a further significant improvement in logical failure rate.
We note that the modification of the surface code, described in Ref.~\cite{Tuckett2018}, simply exchanges the roles of $Z$ and $Y$ operators in stabilizer and logical operator definitions.
Therefore, results for the modified surface code with $Z$-biased noise can equivalently be expressed in terms of the unmodified surface code and $Y$-biased noise, where $Y$ errors occur more frequently than $X$ or $Z$ errors.
In order to frame our analysis in the context of the familiar unmodified surface code and to simplify comparison with other codes, we consider pure $Y$ noise and $Y$-biased noise on the surface code, with $X$- and $Z$-parity checks, throughout this paper.
However, we emphasize that our results apply equally to the modified surface code with pure $Z$ noise or the $Z$-biased noise prevalent in many quantum architectures.

Our main numerical result is to demonstrate that the threshold error rate of the tailored surface code saturates the hashing bound for all biases.
While the numerical results of Ref.~\cite{Tuckett2018} indicate that the threshold error rate of the tailored surface code approaches the hashing bound for low to moderate bias, the threshold estimates fall short for higher and infinite bias.
Using a tensor-network decoder that converges much more strongly with biased noise, we significantly improve on the results of Ref.~\cite{Tuckett2018}.
Our new results are summarized in Fig.~\ref{fig:threshold-v-bias}, providing strong evidence that the hashing bound can be achieved with a tailored surface code.

Our main analytical result is a structural theorem that reveals a hidden concatenated form of the surface code.
We show that, in the limit of pure $Y$ noise, the surface code
can be viewed as a classical concatenated code with two concatenation levels.
The top level contains the so-called cycle code whose parity checks correspond to
cycles in the complete graph. The bottom level contains several
copies of the repetition code. We prove that the cycle code
has an error threshold of $50\%$ and give
an efficient decoding algorithm that achieves this threshold.
As a corollary, we show that the threshold of the surface code with pure $Y$ noise is $50\%$,
thus answering an open question posed in Ref.~\cite{Tuckett2018}.
The concatenated structure described above  is controlled by the parameter $g=\gcd(j,k)$,
where $j$ and $k$ are dimensions of the surface code lattice.
In particular, the top-level cycle code has length $O(g^2)$,
while the bottom-level repetition codes have length $O(jk/g^2)$. 
Two important special cases  are {\em coprime codes} 
and {\em square codes} that have $g=1$ and $g=j=k$, respectively.
Informally, a coprime surface code 
can be viewed as a repetition code, whereas a square
surface code can be viewed as a cycle code (in the limit of pure $Y$ noise).
We also show that a closely-related family of surface codes called \emph{rotated} codes (defined by boundaries formed at \ang{45} relative to the standard surface code family) can also be seen as repetition codes against pure $Y$ noise.
Although the repetition and the cycle codes both have a $50\%$ error threshold,
we argue that the former performs much better in the subthreshold regime.
This result suggests that coprime and rotated surface codes may have an intrinsic advantage in correcting strongly biased noise.

We present further insights into the origins of the ultrahigh thresholds by investigating the form of logical operators.
We show that logical operators consistent with pure $Y$ noise are much rarer and heavier than those 
consistent with pure $X$ or $Z$ noise, and their structure 
depends strongly on the parameter $g$.
In particular, there are $2^{g-1}$ $Y$-type logical operators of which the minimum weight is $(2g-1)(jk/g^2)$, which compares to $2^{j(k-1)}$ $X$-type logical operators of which the minimum weight is $j$.
In the case of coprime codes, there is only one $Y$-type logical operator, and its weight is $jk$.
Hence, the distance of coprime codes to pure $Y$ noise is $O(n)$, whereas for square codes it is $O(\sqrt{n})$.
We extend these results to rotated surface codes.
We find that rotated codes, with odd linear dimensions, have similar features to coprime codes; in particular, they admit only one $Y$-type logical operator, and its weight is $n$.
This result is a further improvement over coprime codes, since rotated surface codes are, in a sense, optimal~\cite{Bombin2007}. That is, they achieve the same distance as standard surface codes with approximately half the number of physical qubits.

Leveraging features of the structure of rotated codes with pure $Y$ noise, we develop a tensor-network decoder that achieves much more strongly converged decoding with $Y$-biased noise compared with the decoder in Ref.~\cite{Bravyi2014} and exact maximum-likelihood decoding in the limit of pure $Y$ noise.

We perform numerical simulations, using exact maximum-likelihood decoding to confirm the $50\%$ threshold for the surface code with pure $Y$ noise and demonstrate a significant reduction in logical failure rate for coprime and rotated codes compared to square codes with pure $Y$ noise.
In particular, we demonstrate that the logical failure rate decays exponentially with the distance to pure $Y$ noise such that a target logical failure rate may be achieved with quadratically fewer physical qubits by using coprime or rotated codes compared with standard (square) surface codes.

Finally, we demonstrate a remarkable property of surface codes: By \emph{removing} approximately half the physical qubits from a square code to yield a rotated code with the same odd linear dimensions, we observe a significant reduction in logical failure rate with biased noise.
Specifically, we perform numerical simulations, using strongly converged approximate maximum-likelihood decoding, to demonstrate the aforementioned significant reduction in logical failure rate against biased noise that is achieved using a rotated $j{\times}j$ code, containing $n=j^2$ physical qubits, compared to a square $j{\times}j$ code, containing $n=2j^2-2j+1$ physical qubits.
Figure~\ref{fig:coprime-v-square-small-code-comparison} summarizes this result, comparing logical failure rate as a function of physical error probability for a rotated $9{\times}9$ code ($81$ qubits) and a square $9{\times}9$ code ($145$ qubits) across a range of biases.
We see that the advantage of the rotated code over the square code is greatest in the limit of pure $Y$ noise ($\eta=\infty$) and remains significant down to a more modest bias, $\eta=100$ (where $Y$ errors are $100$ times more likely than both $X$ and $Z$ errors).
We further argue that, for a given bias, the relative advantage of (odd) rotated codes over square codes increases with code size, until low-rate errors become the dominant source of logical failure and high-rate errors are effectively suppressed, motivating the search for efficient near-optimal biased-noise decoders for rotated codes.

Note that this performance with biased noise is not shared by all topological codes; in stark contrast, the triangular 6.6.6 color code~\cite{Bombin2006} exhibits a decrease in threshold with bias; see Appendix~\ref{sec:color-thresholds}.

The paper is structured as follows.
Section~\ref{sec:definitions} provides some definitions used throughout the paper.
Our main analytical results for surface codes with pure $Y$ noise are in Sec.~\ref{sec:y-noise-theorems}.
Our numerical results for surface codes with pure $Y$ noise and $Y$-biased noise are in Secs.~\ref{sec:performance-pure} and \ref{sec:performance-biased}, respectively.
Section~\ref{sec:tn-decoding-rotated-codes} defines the tensor-network decoder used in simulations of $Y$-biased noise on rotated codes.
We conclude in Sec.~\ref{sec:discussion} with a discussion of our results in the context of prior work and raise some open questions for future work.
Finally, Appendix~\ref{sec:color-thresholds} gives comparative results for color codes, and Appendix~\ref{sec:y-decoder} defines the exact maximum-likelihood decoder used in simulations of pure $Y$ noise on square and coprime surface codes.

\section{\label{sec:definitions}Definitions}

\paragraph*{Standard surface code.---}
We consider $j{\times}k$ standard surface codes~\cite{Bravyi1998} on a square lattice with ``smooth'' top and bottom boundaries and ``rough'' left and right boundaries.
Physical qubits are associated with edges on the lattice.
Following the usual convention, stabilizer generators consist of $X$ operators on edges around vertices, $A_v = \prod_{e \in v} X_e$, and $Z$ operators on edges around plaquettes, $B_p = \prod_{e \in p} Z_e$.
The stabilizer group is, therefore, $\mathcal{G} = \langle A_v, B_p \rangle$.
Up to multiplication by an element of $\mathcal{G}$, the $\overline{X}$ ($\overline{Z}$) logical operator consists of $X$ ($Z$) operators along the left (top) edge, such that $\overline{X}, \overline{Z} \in \mathcal{C(G) \setminus  G}$ and $\overline{X}\overline{Z} = -\overline{Z}\overline{X}$, where $\mathcal{C(G)} = \{ f \in \mathcal{P} : fg=gf\ \forall\ g \in \mathcal{G} \}$ is the centralizer of $\mathcal{G}$ and $\mathcal{P}$ is the group of $n$-qubit Paulis.
As such, a $j{\times}k$ surface code encodes one logical qubit into $n=2jk-j-k+1$ physical qubits with distance $d=\min(j,k)$.
Figure~\ref{fig:surface-code} illustrates a $4{\times}5$ surface code.

\begin{figure}[ht]
  \includegraphics{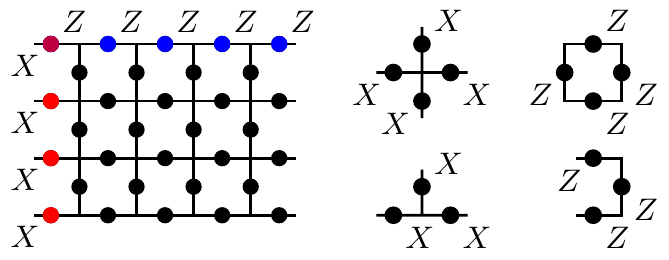}
  \caption{\label{fig:surface-code}
    Standard $4{\times}5$ surface code, with logical operators given by a product of $X$ along the left edge and a product of $Z$ along the top edge.
    Stabilizer generators are shown at the right.
  }
\end{figure}

\paragraph*{Rotated surface code.---}
We also consider rotated surface codes, which are defined by drawing the boundary at \ang{45} relative to the standard surface code lattice~\cite{Bombin2007}; see Fig.~\ref{fig:rotated-code}(a).
As with standard codes, stabilizer generators consist of $X$ ($Z$) operators on edges around vertices (plaquettes), with these restricted to two qubits on the boundaries.
The $\overline{X}$ ($\overline{Z}$) logical operator consists of $X$ ($Z$) operators along the northeast (northwest) edge.
The rotated code is usually, and equivalently, depicted as in Fig.~\ref{fig:rotated-code}(b), where shaded and blank faces correspond to $X$- and $Z$-type stabilizer generators, respectively.
As such, a rotated $j{\times}k$ surface code encodes one logical qubit into $n=jk$ physical qubits with distance $d=\min(j,k)$. Unless otherwise stated, we consider rotated surface codes with $j$ and $k$ odd.

\begin{figure}[ht]
  \includegraphics{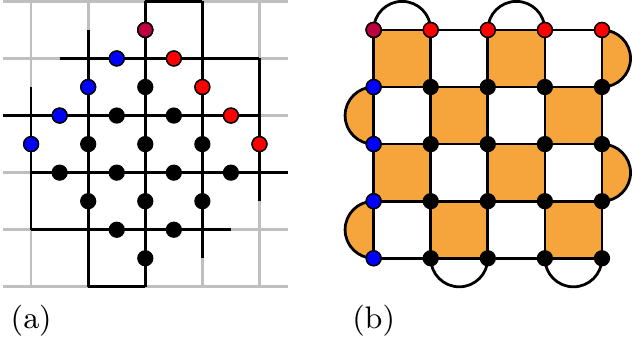}
  \caption{\label{fig:rotated-code}
    (a) Rotated $5{\times}5$ surface code defined by drawing the boundary at \ang{45} relative to the surface code lattice.
    Logical operators are given by a product of $X$ along the northeast edge and $Z$ along the northwest edge.
    As with the standard code, stabilizer generators consist of $X$ ($Z$) operators on edges around vertices (plaquettes).
    (b) Rotated $5{\times}5$ surface code as it is usually, and equivalently, depicted, where shaded (blank) faces corresponding to $X$-type ($Z$-type) stabilizer generators.
  }
\end{figure}

\paragraph*{Surface code families ---}
For standard $j{\times}k$ surface codes, we define the following code families: \emph{square} where $j = k$; $\gcd(j, k)=g$ const; and \emph{coprime} where $g{=}1$ (special case of $g$ constant).
In addition, for rotated $j{\times}k$ surface codes, we define the family of \emph{rotated} codes with $j$ and $k$ odd.

\paragraph*{$Y$-type stabilizers and logical operators.---}
We define a $Y$-type stabilizer to be any operator on a code that is in the stabilizer group $\mathcal{G}$ and consists only of $Y$ and identity single-qubit Paulis.
We define a $Y$-type logical operator to be any operator on a code that is in $\mathcal{C(G) \setminus G}$ and consists only of $Y$ and identity single-qubit Paulis.
We define $X$- and $Z$-type stabilizers and logical operators analogously.
As usual, the weight of an operator is the number of nonidentity single-qubit Paulis applied by the operator.

\paragraph*{$Y$-distance.---}We define $Y$-distance, or distance $d_Y$ to pure $Y$ noise, of a code as the weight of the minimum-weight $Y$-type logical operator.
$X$- and $Z$-distance are defined analogously.
The overall distance of the code is defined in the usual way and is upper bounded by $\min(d_X, d_Y, d_Z)$.

\paragraph*{$Y$-biased noise.---}
Several conventions have previously been used to define biased Pauli noise models~\cite{Ioffe2007, Sarvepalli2009, LaGuardia2014, Robertson2017, Aliferis2008, Aliferis2009, Stephens2013, Stephens2008, Napp2013, Brooks2013, Li2018, Rothlisberger2012, Xu2018, Webster2015, Tuckett2018}.
We adapt the approach of Ref.~\cite{Tuckett2018} to $Y$-biased noise, by considering an independent, identically distributed Pauli noise model defined by an array $\bm{p}=(1-p,p_X,p_Y,p_Z)$ corresponding to the probabilities of each single-qubit Pauli $I$ (no error), $X$, $Y$, and $Z$, respectively, such that the probability of any error on a single qubit is $p = p_X + p_Y + p_Z$.
We define bias $\eta$ to be the ratio of the probability of a $Y$ error to the probability of a non-$Y$ error such that $\eta = p_Y / (p_X + p_Z)$.
For simplicity, we restrict to the case $p_X = p_Z$.
With this definition $\eta=1/2$ corresponds to standard depolarizing noise with $p_X=p_Y=p_Z=p/3$, and the limit $\eta \to \infty$ corresponds to pure $Y$ noise, i.e., only $Y$ errors with probability $p$.
We define $X$- and $Z$-biased noise analogously.

\section{\label{sec:y-noise-theorems}Features of surface codes with pure Y noise}

In this section, we present our analytical results for surface codes with pure $Y$ noise.
In Secs.~\ref{sec:syndromes}--\ref{sec:y-type-logicals}, we present results for standard surface codes, and, in Sec.~\ref{sec:rotated-codes}, we relate these results to rotated surface codes.
We first highlight the specificities of syndromes of pure $Y$ noise.
Our main result reveals that error correction with the standard surface code with pure $Y$ noise is equivalent to a concatenation of two classical codes: the repetition code at the bottom level and the cycle code at the top level.
As a corollary, we show that the surface code with pure $Y$ noise has a threshold of $50\%$.
We also highlight that, for standard $j{\times}k$ surface codes with small $g = \gcd(j,k)$, the more effective repetition code dominates the performance of the code.
We then give explicit formulas for the minimum weight and count of $Y$-type logical operators.
Finally, we relate these results to rotated surface codes.
These results explain the origins of the ultrahigh thresholds of the surface code with $Y$-biased noise, as seen in Ref.~\cite{Tuckett2018} and improved in Sec.~\ref{sec:thresholds-biased}, as well as the lower logical failure rates seen with coprime and rotated surface codes, presented in Secs.~\ref{sec:co-prime-advantage-pure} and \ref{sec:co-prime-advantage-biased}.

\subsection{Syndromes of pure Y noise}
\label{sec:syndromes}

An obvious feature of $Y$ noise on the surface code is that $Y$ errors anticommute with both $X$- and $Z$-type stabilizer generators, providing additional bits of syndrome information.
For comparison, Fig.~\ref{fig:syndromes} shows a sample of $Y$-error configurations alongside identically placed $X$- and $Z$-error configurations with corresponding anticommuting syndrome locations for each error type.
In each case, we see that $Y$-error strings anticommute with more syndrome locations than $X$- or $Z$-error strings, providing the decoder with more information about the location of errors to be corrected.

\begin{figure}[ht]
  \includegraphics{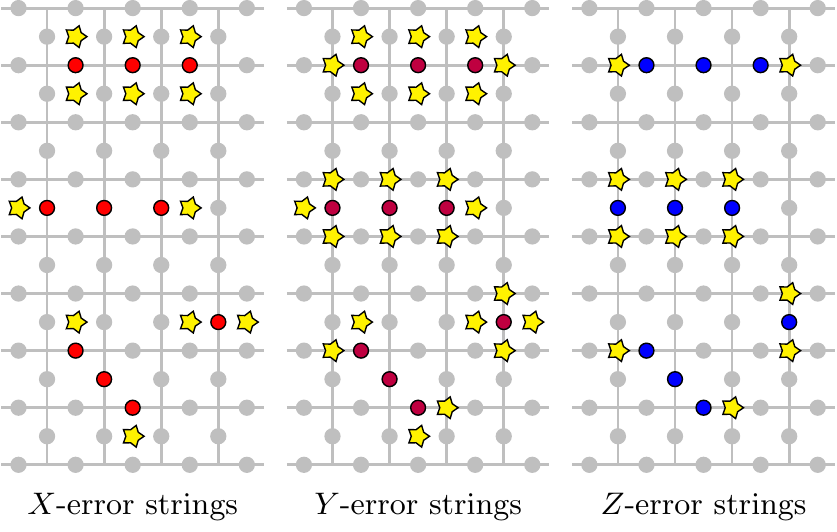}
  \caption{\label{fig:syndromes}
    A sample of $X$-, $Y$-, and $Z$-error strings, indicated by colored circles, with corresponding anticommuting syndrome locations, indicated by yellow stars.
  }
\end{figure}

We remark that the displacement between the $X$- and $Z$-type stabilizer generators appears to be significant.
For example, the color 6.6.6 code has colocated $X$- and $Z$-type stabilizer generators, so that, even if $Y$ errors anticommute with more stabilizer generators, the number of distinct syndrome locations triggered by $Y$ errors is no greater than for $X$ or $Z$ errors.

\subsection{Structure of the standard surface code with pure Y noise}
\label{sec:code-equivalence}

In this section, we consider standard surface codes subject to pure $Y$ noise.
We describe a polynomial-time decoding algorithm and prove that it
achieves an error threshold of $50\%$.
We also derive an exponential upper
bound on the probability of logical errors in the subthreshold regime.
Our main result is a structural theorem that reveals a hidden concatenated
structure of the surface code
and highlights the role of the parameter $g=\gcd{(j,k)}$.
The theorem implies that error correction with the surface code
subject to $Y$ noise
can be viewed as a concatenation of two classical codes:
the repetition code at the bottom level and the so-called cycle
code at the top level. Both codes admit efficient
decoding algorithms and have an error threshold of $50\%$,
although the repetition code scores much better in terms of the logical error probability.
We show that, for a fixed number of qubits,
the size of each code can vary drastically depending on the
value of $g$.
Loosely speaking, the error-correction workload is shared between the two codes
such that for small $g$ the dominant contribution comes from the more effective
repetition code, which explains the enhanced performance of coprime surface codes ($g=1$)
observed in the numerics.

\subsubsection{Concatenated structure}
\label{sec:concatenated-structure}

Consider a Pauli error
\begin{equation}
\label{P(y)}
P(y)\equiv Y_1^{y_1} \otimes Y_2^{y_2} \otimes \cdots \otimes Y_n^{y_n},
\end{equation}
where $y\in \{0,1\}^n$. As described in Sec.~\ref{sec:syndromes},
the syndrome of $P(y)$ is given by
\begin{equation}
\label{ab}
A_v(y)=\sum_{e\in v} y_e  \qquad \mbox{and} \qquad B_p(y)= \sum_{e\in p} y_e
\end{equation}
where $v$ and $p$ run over all vertices and all plaquettes of the lattice
and the sums are modulo two.
A decoding algorithm takes as input the error syndrome
and outputs a candidate recovery operator $P(y')$
that agrees with the observed syndrome.
The decoding succeeds if $y'=y$ and fails
otherwise. [More generally,
the decoder needs to identify only the
equivalence class of errors that contains $P(y)$,
where the equivalence is defined modulo stabilizers of the surface code.]

Consider a classical linear code of length $n$
defined by the parity checks  $A_v(y)=0$ and $B_p(y)=0$ for all $v$ and $p$.
We shall refer to this code as a {\em Y-code}.
As described above, error correction for the surface code
subject to $Y$-noise is equivalent to error correction for the $Y$-code
subject to classical bit-flip errors.
We now establish the structure of the $Y$-code.
For any integer $m\ge 3$,
let $K_m$ be the complete graph with $m$ vertices
and $e=m(m-1)/2$ edges. Consider bit strings
$x\in \{0,1\}^e$ such that bits of $x$ are associated with edges of the
graph $K_m$. Let $x_{i,j}$ be the bit associated with an edge $(i,j)$.
Here it is understood that $x_{i,j}=x_{j,i}$.
Define a {\em cycle code} $\calC_m$ of order $m$ that encodes $m-1$ bits into $e$ bits
with parity checks
\be
\label{parity-checks}
x_{i,j} \oplus x_{j,k} \oplus x_{i,k} = 0 \qquad \mbox{for all $1\le i<j<k\le m$}.
\ee
Thus, parity checks of $\calC_m$ correspond to cycles (triangles) in the graph $K_m$.
Note that Eq.~(\ref{parity-checks}) defines a redundant set of parity checks.
It is well known that any connected graph with $m$ vertices and $e$ edges has
$e-m+1$ independent cycles. Thus, $\calC_m$ has $e-(m-1)$ independent parity checks.
The number of encoded bits is  $m-1$.
Note that $\calC_2$ is a trivial code (it has no parity checks).
Let $\REP{(m)}$ be the repetition code that encodes one bit into $m$ bits.
We can now describe the structure of the $Y$-code.
\begin{theorem}[$Y$-code structure]\label{thm:code-equivalence}
The $Y$-code is a concatenation of the
cycle code $\calC_{g+1}$ at the top level
and $g(g+1)/2$ repetition codes at
the bottom level.  The latter consists of
repetition codes
$\REP{(jk/g^2)}$, $\REP{(2jk/g^2)}$, and
$\REP{(4jk/g^2)}$ with multiplicities
$1$, $2(g-1)$, and $g(g+1)/2 - 2g+1$, respectively.
\end{theorem}
An important corollary of the theorem is that a
decoding algorithm for the cycle code
can be directly applied to correcting $Y$ errors in the surface code.
Indeed, a decoder for the $Y$-code can be  constructed in a level-by-level fashion
such that the bottom-level repetition codes are decoded first and
the top-level cycle code is decoded afterwards.

For example, Theorem~\ref{thm:code-equivalence} implies that, with pure $Y$ noise,
a coprime ($g=1$) surface code is essentially a single repetition code of a size growing linearly with $n$, whereas a square surface code is equivalent to the concatenation of bottom-level fixed-size repetition codes
$\REP{(1)}$, $\REP{(2)}$, and
$\REP{(4)}$
 and a top-level cycle code of a size growing linearly with $n$,
where $n$ is the number of physical qubits in the surface code.

\begin{proof}
\begin{figure}[ht]
\includegraphics[width=\linewidth]{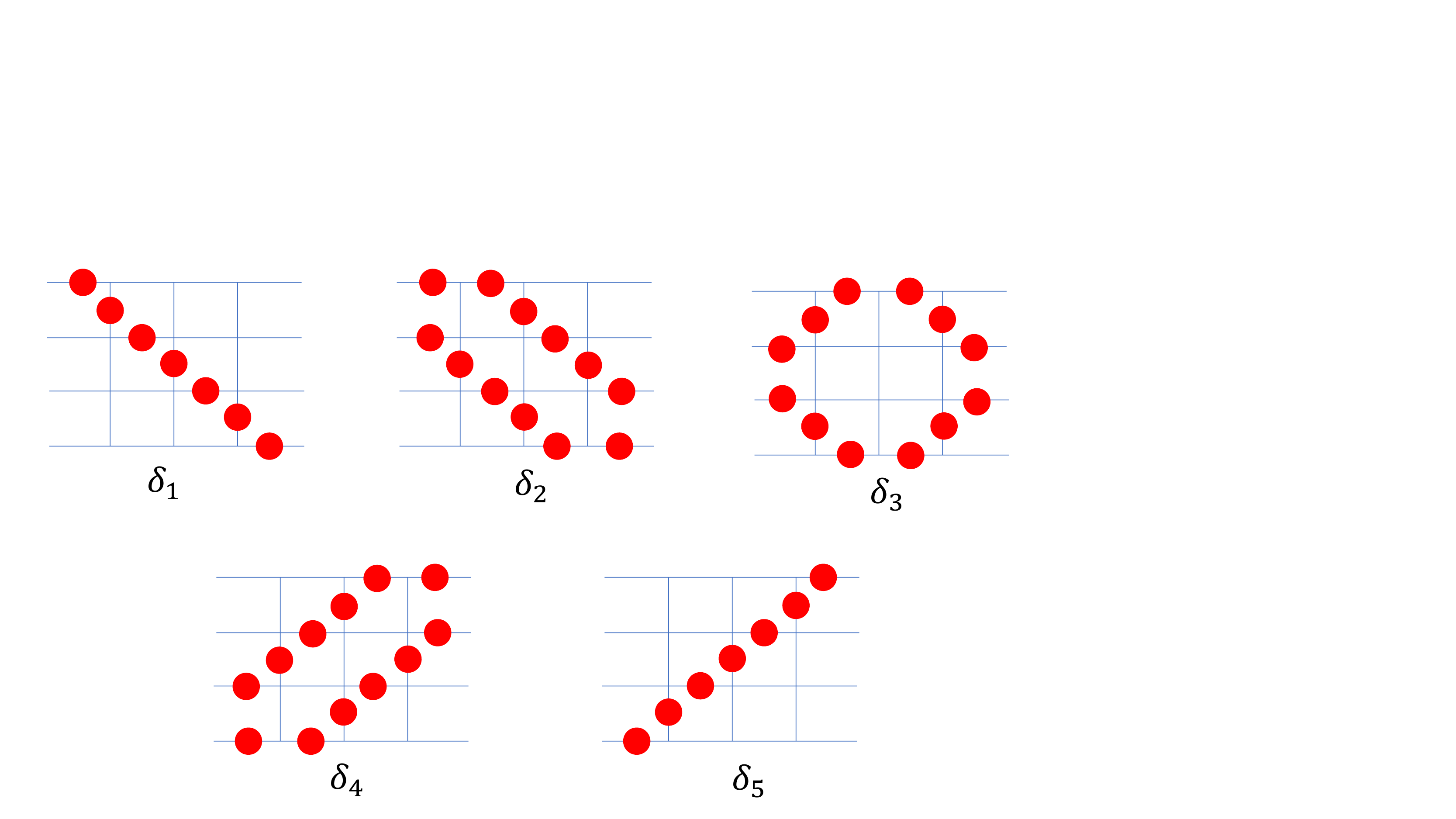}
\caption{Diagonals $\delta^i$ for the $4\times 4$ surface code.
We consider the symmetry group $\calR$ generated
by reflections of the lattice against $\delta^1$ and $\delta^5$.
Note that any diagonal $\delta^i$ is symmetric under
reflections from $\calR$.
\label{fig:diagonals}}
\end{figure}
Let us first prove the theorem in the special case 
of square surface codes,  $j=k=g$.
Let $\calG\subset \{0,1\}^n$ be the code space of the $Y$-code.
We use a particular basis set of codewords
called {\em diagonals}.
The $j\times j$ lattice has $j+1$ diagonals denoted
$\delta^1,\delta^2,\ldots,\delta^{j+1}\in \calG$;
see Fig.~\ref{fig:diagonals}.
Given a codeword $y\in \calG$,
let $\partial y \in \{0,1\}^j$ be the restriction of $y$ onto the top horizontal row of edges
in the surface code lattice.
We claim that $y$  is uniquely determined by $\partial y$.
Indeed, let $H_1,\ldots,H_j$ be the rows of horizontal edges
(counting from the top). Let $V_2,\ldots,V_j$ be the rows
of vertical edges (counting from the top).
By definition, the restriction of $y$ onto $H_1$ coincides with $\partial y$.
Suppose the restriction of $y$ onto $H_1V_2\ldots H_p$ is already determined
(initially $p=1$). Vertex  parity checks $A_v(y)=0$ 
located at the row $H_p$
then determine the restriction of $y$ onto $V_{p+1}$.
Likewise, suppose the restriction of $y$ onto $H_1V_2\ldots H_pV_p$
is already determined. Plaquette parity checks
$B_p(y)=0$ located at the row $V_p$ then determine the restriction of
$y$ onto $H_{p+1}$. Proceeding inductively shows that 
any codeword $y\in \calG$ is uniquely determined by $\partial y$.

Define  bit strings
\[
e^1=100\ldots0, \quad e^2=010\ldots 0, \quad  e^3=001\ldots 0 \quad \mbox{etc}.
\]
Then  $\partial \delta^1=e^1$, $\partial \delta^i = e^{i-1} + e^i$ for
$2\le i \le j$, and $\partial \delta^{j+1}=e^j$; see Fig.~\ref{fig:diagonals}.
It follows that $\partial \delta^1,\ldots,\partial \delta^j$ span the
binary space $\{0,1\}^j$. Accordingly,
the diagonals $\delta^1,\ldots,\delta^j$ span the
code space $\calG$ and
\[
\delta^{j+1}=\delta^1\oplus \delta^2 \oplus  \cdots \oplus  \delta^j.
\]
In particular, $\dim{(\calG)}=j$, that is,
the $Y$-code encodes $j$ bits into $n$ bits.
\begin{figure}[ht]
\includegraphics[width=3cm]{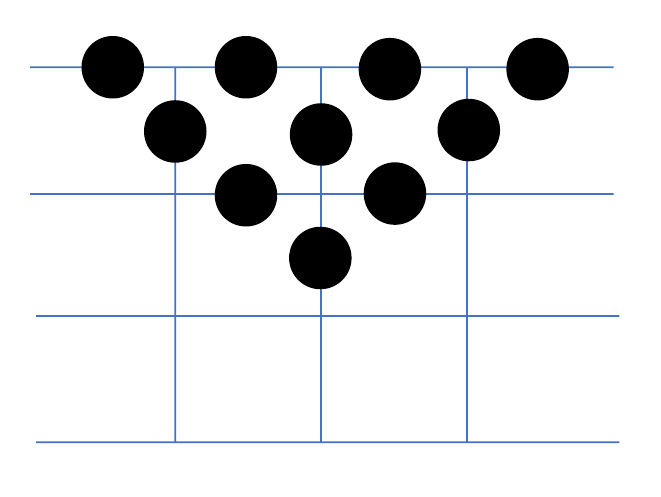}
\caption{A set of qubits $\calO$ such that
each orbit of $\calR$ contains exactly one qubit from $\calO$.
In this example
the group $\calR$
has ten orbits of size 1, 2, and 4.
\label{fig:orbits}}
\end{figure}
\begin{figure}[ht]
\includegraphics[width=\linewidth]{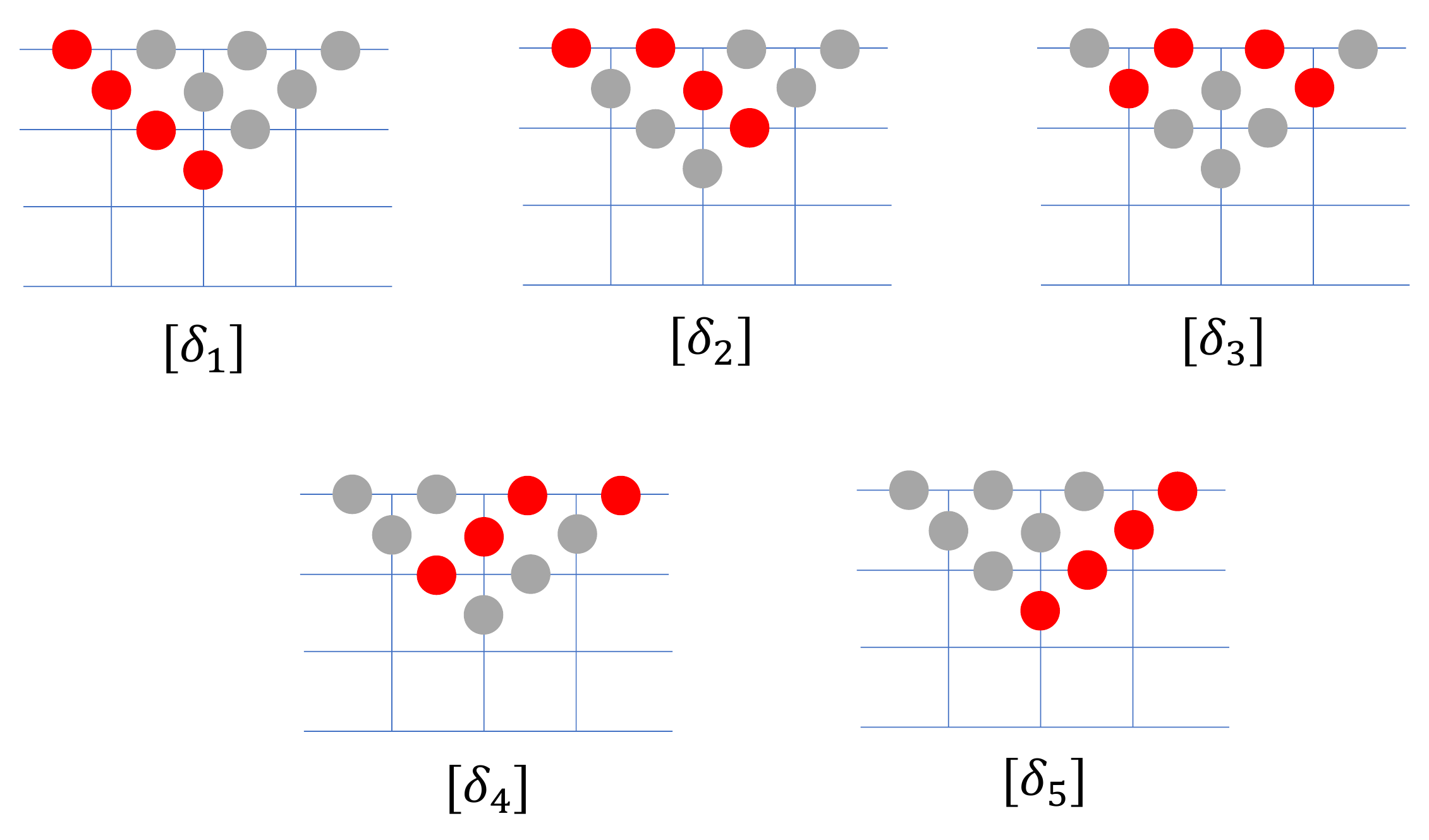}
\caption{Restrictions of the diagonal
$\delta^i$ onto $\calO$ define a basis
set of codewords for the top-level code.\label{fig:diagonals1}}
\end{figure}

Let $\calR\cong \ZZ_2\times \ZZ_2$ be
a group generated by reflections of the lattice
against the diagonals $\delta^1$ and $\delta^{j+1}$.
Note  that
any diagonal $\delta^i$ is invariant under reflections
from $\calR$; see  Fig.~\ref{fig:diagonals}.
Suppose $f$ is an edge of the surface code lattice.
Let $\calR(f)$ be the orbit of $f$ under the action of $\calR$.
The above shows that any diagonal $\delta^i$
is constant on orbits of $\calR$; that is,
$\calR(f)=\calR(g)$ implies that $\delta^i_f=\delta^i_g$.
Since the diagonals $\delta^i$ span the full code space
$\calG$, we conclude that any codeword $y\in \calG$
is constant on orbits of $\calR$; that is,
$\calR(f)=\calR(g)$ implies that $y_f=y_g$.
Equivalently, each orbit of $\calR$
of size $m$  gives rise to the repetition code
$\REP(m)$. A simple counting shows that $\calR$ has
a single orbit of size $1$ (the central vertical edge) and 
$2(j-1)$ orbits of size $2$ (pairs of qubits located on
the diagonals $\delta^1$ and $\delta^{j+1}$),
whereas all remaining orbits have size $4$, which proves the last statement of the theorem
(in the special case $j=k$).

Fix a set of qubits $\calO$  such that
each orbit of $\calR$ contains exactly one qubit from $\calO$.
In other words, $\calO$ is a set of orbit representatives.
We choose $\calO$ as shown
in Fig.~\ref{fig:orbits}.
A simple counting shows that $|\calO|=j(j+1)/2$.
Consider a codeword $y\in \calG$ and let
$[y]\in \{0,1\}^{|\calO|}$ be a vector obtained by
restricting $y$ onto $\calO$.
We define the top-level code as a linear
subspace $\calL \subseteq \{0,1\}^{|\calO|}$
spanned by vectors $[y]$ with $y\in \calG$.
Equivalently, $\calL$ is spanned
by vectors $[\delta^i]$ with $i=1,\ldots,j+1$.
A direct inspection shows that each qubit
$e\in \calO$ belongs to exactly two vectors
$[\delta^i]$ and $[\delta^k]$ for some $i\ne k$;
see Fig.~\ref{fig:diagonals1} for an example.
Thus, one can identify $\calO$ with the set of edges
of the complete graph $K_{j+1}$, whereas
the vectors $[\delta^i]$ can be identified
with ``vertex stabilizers" in $K_{j+1}$.
In other words, the support of each vector $[\delta^i]$
coincides with the set of edges incident to some vertex of $K_{j+1}$.
We conclude that parity checks of $\calL$ correspond to
closed loops in $K_{j+1}$. Thus, the top-level code
coincides with the cycle code $\calC_{j+1}$.

The above proves the theorem in the special case $j=k$.
Consider now the general case $j\ne k$.
Let us tile the surface code lattice by $t=jk/g^2$  tiles
of size $g\times g$ as
shown in Fig.~\ref{fig:tiles}. Note that each horizontal edge 
is fully contained in  some  tile. Let us say that a vertical edge is 
a boundary edge if it overlaps with the boundary of some adjacent tiles.
If one ignores the boundary edges, each tile contains a single copy
of the $g\times g$ surface code.
\begin{figure}[ht]
\includegraphics[width=7cm]{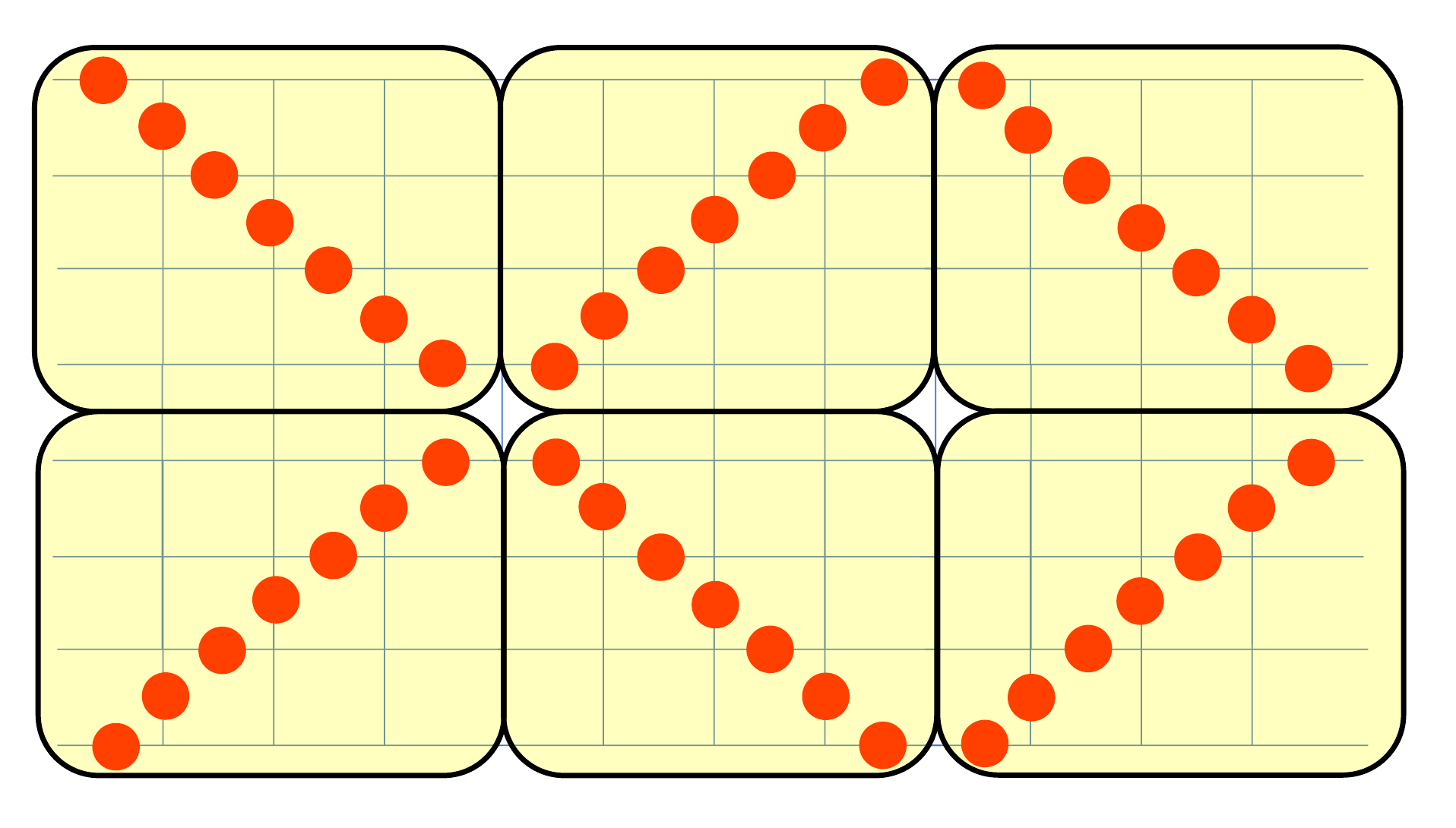}
\caption{Partition of the $8\times 12$ surface code
into $4\times 4$ tiles. Solid red circles:  The extended
diagonal $\Delta^1$ alternating between $\delta^1$ and
$\delta^5$; see Fig.~\ref{fig:diagonals}.
\label{fig:tiles}}
\end{figure}
For each tile, define the diagonals $\delta^1,\delta^2,\ldots,\delta^{g+1}$ as above.
Let $\calG$ be the code space of the $Y$-code for the full $j\times k$ lattice.
Recall that any codeword $y\in \calG$
is fully determined by its projection $\partial y$ onto the top horizontal row of edges.
Using this property, one can easily verify that the code space $\calG$ 
is spanned by ``extended diagonals" $\Delta^i$
such that the restriction of 
$\Delta^i$ onto the top-left tile coincides with $\delta^i$ and
$\Delta^i$ alternates between $\delta^i$ and $\delta^{g+2-i}$
in a checkerboard fashion; see Fig.~\ref{fig:tiles1}.
\begin{figure}[ht]
\includegraphics[width=7cm]{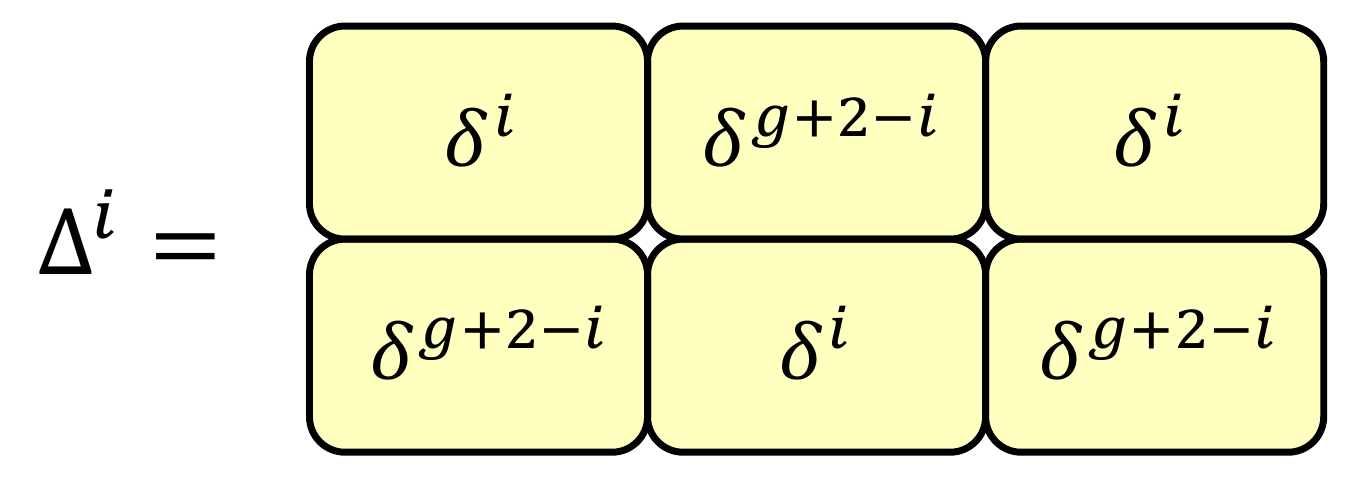}
\caption{Extended diagonal $\Delta^i$.
\label{fig:tiles1}}
\end{figure}
An example of the extended diagonal $\Delta^1$ is shown in Fig.~\ref{fig:tiles}.
By definition, $\Delta^i$ has no support on the boundary edges, which implies that the $Y$-code has a weight-$1$ parity check for each boundary edge.
Ignoring such weight-$1$  checks,
each codeword $\Delta^i$ consists of $t$ copies of the diagonal $\delta^i$ with
some copies being reflected. Considering $t$ copies of each codeword
instead of a single copy is   equivalent to replacing the repetition codes $\REP{(1)}$, $\REP{(2)}$, and 
$\REP{(4)}$ in the above analysis by
$\REP{(t)}$, $\REP{(2t)}$, and 
$\REP{(4t)}$, respectively, where $t=jk/g^2$ is the number of tiles.
\end{proof}

\subsubsection{Decoding the cycle code}
\label{sec:cycle_decoding}

Here, we consider the cycle code subject to random  errors.
We give a polynomial-time decoding algorithm that achieves the error threshold
of $50\%$.
Fix some integer $m\ge 3$
and consider the cycle code $\calC_m$ defined in Sec.~\ref{sec:concatenated-structure}.
Recall that $\calC_m$ has length $n=m(m-1)/2$.
We consider independent and identically distributed (IID) bit-flip errors such that each bit is flipped with probability
$p\in [0,1/2)$. Define an error bias $\epsilon>0$ such that
\begin{equation}
\label{bias}
2p(1-p)=\frac12- \epsilon.
\end{equation}
\begin{lemma}[Cycle code decoder]
\label{lemma:decoder}
Let  $e\in \{0,1\}^n$ be a random IID error with a bias $\epsilon$.
There exists an algorithm
that takes as input the  syndrome of $e$
and outputs a bit string $e'\in \{0,1\}^n$ such that
\be
\label{upper_bound}
\mathrm{Prob}[e'=e] \ge 1- 2m^2 \cdot \exp{(-2\epsilon^2 m)}.
\ee
The algorithm has runtime $O(m^3)$.
\end{lemma}
\begin{proof}
Recall that the cycle code $\calC_m$ is defined on the complete graph
with $m$ vertices such that each bit of $\calC_m$ is located on
some edge $(i,j)$ of the graph.
Let $e_{i,j}$ be the error bit associated with an edge $(i,j)$.
We begin by giving a subroutine that identifies a single error bit $e_{i,j}$.
Without loss of generality, consider the edge $(1,2)$.
This edge is contained in $m-2$ triangles that give rise to syndrome bits
\begin{align}
\label{triangles}
s_3 &=  e_{1,2} \oplus e_{2,3} \oplus e_{3,1}, \nonumber \\
s_4 &=e_{1,2} \oplus e_{2,4} \oplus e_{4,1},  \nonumber \\
&\cdots  \nonumber \\
s_m & =  e_{1,2} \oplus e_{2,m} \oplus e_{m,1}.
\end{align}
Since errors on different edges of each triangle are independent,
the conditional probability distributions of syndromes $s_j$ for a given
error bit $e_{1,2}$ are
\begin{align*}
\mathrm{Prob}[s_j=1|e_{1,2}=0] =  \frac12 -\epsilon, \\
\mathrm{Prob}[s_j=0|e_{1,2}=0] = \frac12  + \epsilon, \\
\mathrm{Prob}[s_j=1|e_{1,2}=1] = \frac12 + \epsilon, \\
\mathrm{Prob}[s_j=0|e_{1,2}=1] = \frac12  -\epsilon.
\end{align*}
Furthermore, since different triangles in Eq.~(\ref{triangles})  intersect
only on the edge $(1,2)$, we have
\be
\mathrm{Prob}[s_3,\ldots,s_m|e_{1,2}]=\prod_{j=3}^m \mathrm{Prob}[s_j|e_{1,2}].
\ee
This equation is an IID distribution of $m-2$ bits
which is $\epsilon$ biased toward $e_{1,2}$.
 Hoeffding's inequality  gives
\[
\mathrm{Prob}[s_3+ \ldots + s_m \ge m/2|e_{1,2}=0]\le 4  \exp{(-2\epsilon^2 m)}
\]
and
\[
\mathrm{Prob}[s_3+ \ldots + s_m \le m/2|e_{1,2}=1]\le 4\exp{(-2\epsilon^2 m)}.
\]
The desired subroutine  outputs $e_{1,2}=0$ if
$s_3+ \ldots + s_m \le m/2$ and $e_{1,2}=1$ otherwise.
Clearly, the above calculations take time $O(m)$.

The full decoding algorithm applies the above subroutine
independently to each edge of the graph learning error
bits one by one.
By the union bound, such an algorithm misidentifies the error with
a probability of at most $2m^2 \exp{(-2\epsilon^2 m)}$
since the complete graph $K_m$ has $m(m-1)/2$ edges.
The overall runtime of the algorithm is $O(m^3)$.
\end{proof}

Note that the decoding algorithm of Lemma~\ref{lemma:decoder}
can be viewed as a single round of the standard belief-propagation algorithm,
which is commonly used to decode classical low-density parity check (LDPC) codes.
Also recall that the cycle code $\calC_m$ has length $n\sim m^2/2$.
Thus, the probability of a logical error in Eq.~(\ref{upper_bound})
decays exponentially with $\sqrt{n}$ [this scaling
is unavoidable, since the cycle code $\calC_m$ has distance $O(m)$].
As a consequence, the proposed decoder
performs very poorly in the small-bias regime. For example,
reducing the error rate from $49\%$ to $1\%$ requires code
length $n\approx 10^{17}$ [here, we use Eq.~(\ref{upper_bound})
as a rough estimate of the logical error probability].
In contrast, the logical error probability of the repetition
code $\REP{(n)}$ decays exponentially with $n$.

\subsection{\label{sec:y-threshold}Threshold of the standard surface code with pure Y noise}

The standard surface code with pure $Y$ noise is equivalent to a concatenation of two classical codes, as shown above, and both of these classical codes have thresholds of $50\%$.
These results lead directly to the fact that the threshold of the surface code with pure $Y$ noise is $50\%$.
Indeed, let us employ the level-by-level decoding strategy such that the bottom-level repetition codes are decoded first.
Assume that the pure $Y$ noise has error rate $p<1/2$.
Then, the $j$th repetition code makes a logical error with probability $p_j\le p<1/2$.
The effective error model for the top-level cycle code is a product of symmetric binary channels with error rates $p_1,\ldots,p_m\le p$, where $m=g(g+1)/2$ is the length of the cycle code.
One can easily verify that the decoder of Lemma~\ref{lemma:decoder} corrects such a random error with a probability given by Eqs.~(\ref{bias}) and (\ref{upper_bound}).
Finally, Theorem~\ref{thm:code-equivalence} implies that each parity check of the repetition or the cycle code is a linear combination (modulo two) of the plaquette and vertex parity checks of Eq.~(\ref{ab}).
The coefficients in this linear combination can be found by solving a suitable system of linear equations in time $O(n^3)$, which 
enables an efficient conversion between the surface code syndrome and the syndromes of the bottom-level and the top-level code.
To conclude, Theorem~\ref{thm:code-equivalence} and Lemma~\ref{lemma:decoder} have the following corollary.
\begin{corollary}[$Y$-threshold]\label{cor:threshold}
  The error-correction threshold for the surface code with pure $Y$ noise is $50\%$.
  This error threshold can be achieved by a polynomial-time decoding algorithm.
\end{corollary}

In Sec.~\ref{sec:rotated-codes}, we show that the above corollary also applies to rotated surface codes, with odd linear dimensions.
A numerical demonstration of the $50\%$ threshold of the surface code with pure $Y$ noise is given in Sec.~\ref{sec:co-prime-advantage-pure}.

\subsection{Y-type logical operators of the standard surface code}
\label{sec:y-type-logicals}

The structure of standard surface codes with pure $Y$ noise, described in Sec.~\ref{sec:code-equivalence}, also manifests itself in the structure and, consequently, the minimum weight and count of $Y$-type logical operators, i.e., logical operators consisting only of $Y$ and identity single-qubit Paulis.
In this section, we give explicit formulas for the minimum weight and count of $Y$-type logical operators.
Highlighting the cases of coprime and square codes, as well as comparing the formulas to those for $X$- and $Z$-type logical operators, we remark on how the minimum weight and count of $Y$-type logical operators contribute to the performance advantage with pure $Y$ noise and $Y$-biased noise seen in Ref.~\cite{Tuckett2018} and Sec.~\ref{sec:thresholds-biased}, for surface codes, in general, and in Secs.~\ref{sec:co-prime-advantage-pure} and \ref{sec:co-prime-advantage-biased}, for coprime and rotated codes, in particular.

\subsubsection{Logical operator minimum weight}
We show that the minimum-weight $Y$-type logical operator on standard surface codes is comparatively heavy.
The $X$-distance $d_X$ of a code is the weight of the minimum-weight $X$-type logical operator.
Clearly, the minimum-weight $X$-type logical operator on a $j{\times}k$ code is a full column of $X$ operators on horizontal edges, and, hence, $d_X=j$; similarly, $d_Z=k$.
It is also clear that the minimum-weight $Y$-type logical operator on a square $j{\times}j$ code is a full diagonal of $Y$ operators, and, hence, $d_Y=2j-1$.
From the proof of Theorem~\ref{thm:code-equivalence}, it is apparent that, in the case of pure $Y$ noise, a $j{\times}k$ surface code can be viewed as a tiling of $jk/g^2$ copies of a square $g{\times}g$ code, where $g = \gcd(j,k)$.
Therefore, the $Y$-distance of a $j{\times}k$ surface code is given by the following corollary.

\begin{corollary}[$Y$-distance]\label{cor:distance}
  For a standard $j{\times}k$ surface code, the weight of the minimum-weight $Y$-type logical operator, and, hence, the distance of the code to pure $Y$ noise, is
  $$
  d_Y = \frac{(2g-1)jk}{g^2}
  $$
  where $g = \gcd(j,k)$.
\end{corollary}

As shown in Sec.~\ref{sec:rotated-codes}, the $Y$-distance of the rotated $j{\times}k$ surface code, with $j$ and $k$ odd, is $d_Y=jk$.
The distances to pure noise for various surface code families are summarized in Table~\ref{table:distance}.
We note that, for all code families, $Y$-distance exceeds $X$- or $Z$-distance, which is consistent with the increase in error threshold of surface codes with biased noise seen in Ref.~\cite{Tuckett2018} and Sec.~\ref{sec:thresholds-biased}.
Furthermore, we note that the $Y$-distance of square codes is $d_Y=O(\sqrt{n})$ while that of coprime and rotated codes is $d_Y=O(n)$, where $n$ is the number of physical qubits.
This feature of near-optimal and optimal $Y$-distance contributes to the significant improvement in logical failure rate of coprime and rotated codes over square codes with pure $Y$ noise and $Y$-biased noise, see Secs.~\ref{sec:co-prime-advantage-pure} and \ref{sec:co-prime-advantage-biased}.

\begin{table}[ht]
  \caption{\label{table:distance}
    Distances to pure noise for $j{\times}k$ surface code families.
    ($d_P$ refers to the distance to pure $P$ noise, where $P \in \{X, Y, Z\}$.)
  }
  \begin{ruledtabular}
    \begin{tabular}{l l l l}
      Code family & $d_X$ & $d_Y$ & $d_Z$ \\
      \hline
      Square & $j$ & $2j-1$ & $j$ \\
      Coprime & $j$ & $jk$ & $k$ \\
      $\gcd(j,k)=g$ & $j$ & $(2g-1)(jk/g^2)$ & $k$ \\
      Rotated ($j,k$ odd) & $k$ & $jk$ & $j$ \\
    \end{tabular}
  \end{ruledtabular}
\end{table}

\subsubsection{Logical operator count}
We show that $Y$-type logical operators on standard surface codes are comparatively rare.
The number $c_X$ of $X$-type logical operators is equal to the number of ways the logical $\overline{X}$ operator can be deformed by $X$-type stabilizer generators.
The number of $X$-type stabilizer generators (i.e., vertices) on a $j{\times}k$ surface code is $j(k-1)$, and, hence, $c_X=2^{j(k-1)}$; similarly, $c_Z=2^{(j-1)k}$.
From the proof of Theorem~\ref{thm:code-equivalence}, it is apparent that the $g$ basis codewords of the $Y$-code correspond to a single logical operator and a full set of $g-1$ linearly independent $Y$-type stabilizers of a $j{\times}k$ surface code, where $g = \gcd(j,k)$.
Therefore, the number of $Y$-type logical operators of a $j{\times}k$ surface code is given by the following corollary.

\begin{corollary}[$Y$-count]\label{cor:count}
  For a standard $j{\times}k$ surface code, the number of $Y$-type logical operators is
  $$
  c_Y = 2^{g-1}
  $$
  where $g = \gcd(j,k)$.
  The number of $Y$-type stabilizers is also $c_Y$.
\end{corollary}

As shown in Sec.~\ref{sec:rotated-codes}, the number of $Y$-type logical operators on the rotated $j{\times}k$ surface code, with $j$ and $k$ odd, is $c_Y=1$.
The counts of pure noise logical operators for various surface code families are summarized in Table~\ref{table:count}.
We note that, for all code families, the number of logical operators for pure $Y$ noise is much lower than the number for pure $X$ or $Z$ noise, which is consistent with the increase in error threshold of surface codes with biased noise seen in Ref.~\cite{Tuckett2018} and Sec.~\ref{sec:thresholds-biased}.
Furthermore, we note that the number of $Y$-type logical operators for square codes is $c_Y=O(2^{\sqrt{n}})$, while for coprime and rotated codes it is $c_Y=O(1)$, where $n$ is the number of physical qubits.
This feature, as an extreme example of the role of entropy in error correction~\cite{Beverland2019}, contributes to the significant improvement in logical failure rate of coprime and rotated codes over square codes with pure $Y$ noise and $Y$-biased noise, see Secs.~\ref{sec:co-prime-advantage-pure} and \ref{sec:co-prime-advantage-biased}.

\begin{table}[ht]
  \caption{\label{table:count}
    Counts of pure noise logical operators for $j{\times}k$ surface code families.
    ($c_P$ refers to the number of $P$-type logical operators, where $P \in \{X, Y, Z\}$.)
  }
  \begin{ruledtabular}
    \begin{tabular}{l l l l}
      Code family & $c_X$ & $c_Y$ & $c_Z$ \\
      \hline
      Square & $2^{j^2-j}$ & $2^{j-1}$ & $2^{j^2-j}$ \\
      Coprime & $2^{j(k-1)}$ & $1$ & $2^{(j-1)k}$ \\
      $\gcd(j,k)=g$ & $2^{j(k-1)}$ & $2^{g-1}$ & $2^{(j-1)k}$ \\
      Rotated ($j,k$ odd) & $2^{(j-1)(k+1)/2}$ & $1$ & $2^{(j+1)(k-1)/2}$ \\
    \end{tabular}
  \end{ruledtabular}
\end{table}

\subsection{Rotated surface codes}
\label{sec:rotated-codes}

We can relate the results from the previous subsections to rotated surface codes as depicted in Fig.~\ref{fig:rotated-code}.
In particular, we show that rotated codes, with odd linear dimensions, have similar features to coprime codes as given by Corollaries~\ref{cor:distance} and \ref{cor:count}; that is, such rotated codes also admit a single $Y$-type logical operator of weight $O(n)$, where $n$ is the number of physical qubits.
Equivalently, the $Y$-distance of such rotated codes, like coprime codes, is $d_Y=O(n)$; notably, the rotated code is optimal, in that it achieves $d_Y=n$ precisely.
Rotated surface codes with even linear dimensions do not share these features, having distance $d_Y=O(\sqrt{n})$ with pure $Y$ noise, and we do not discuss them further.
We conclude by showing that the $50\%$ threshold of surface codes with pure $Y$ noise, given by Corollary~\ref{cor:threshold}, also applies to (odd) rotated codes.

We consider the rotated surface code, with odd linear dimensions, and two-qubit (four-qubit) stabilizer generators on the boundary (in the bulk), as illustrated in Fig.~\ref{fig:rotated-code}.
The following theorem shows that this version of the surface code is nondegenerate and has a distance of $d_Y=n$ against pure $Y$ noise.
\begin{theorem}[Rotated code $Y$-logical]\label{thm:rotated-y-logical}
  For a rotated surface code, with odd linear dimensions, $Y^{\otimes n}$ is the only nontrivial $Y$-type logical operator, where $n$ is the number of physical qubits.
\end{theorem}

\begin{proof}
  It is clear that $Y^{\otimes n}$ is a $Y$-type logical operator. We show that it is the only nontrivial $Y$-type operator that commutes with every stabilizer of the code. Let $A=\bigotimes_i Y^{\alpha_i}$ be a $Y$-type operator. Consider a row of stabilizer generators (checks) and qubits that are adjacent to this row, numbered as shown below:
  \begin{equation*}
  \includegraphics[width=0.3\textwidth]{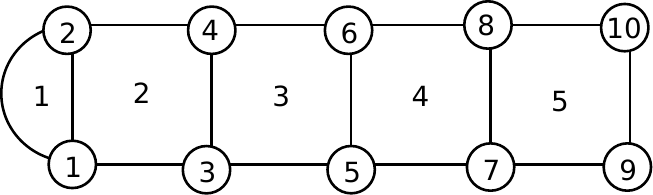}
  \end{equation*}
  
  In order for $A$ to commute with check 1 we require $\alpha_1=\alpha_2$.
  With the parity of these checks determined, we then see that, in order for $A$ to commute with check 2, we need $\alpha_3=\alpha_4$.
  Continuing along the row, we see that every pair of qubits $i,j$ in the same column must satisfy $\alpha_i=\alpha_j$.
  The assumption that the code has odd linear dimensions implies that each row and each column of checks includes a weight-two check, as depicted, ensuring that the same argument can equally be applied to every row or column of checks.
  Therefore, in order for $A$ to commute with all checks, we require $\alpha_1=\alpha_j$ for all $j$; i.e., a nontrivial $Y$-type logical must act as $Y$ on every qubit.
\end{proof}

We note that both the coprime $j{\times}k$ code and the (odd) rotated $j{\times}k$ code are nondegenerate against pure $Y$ noise and have $Y$-distance $d_Y=jk=O(n)$.
However, the rotated code is known to be the optimal layout for surface codes in terms of minimum distance~\cite{Bombin2007}, and this statement is also true in terms of $Y$-distance.
The rotated code has $d_Y=jk=n$, whereas the coprime code has $d_Y=jk=O(n)$ but contains $n=2jk-j-k+1$ physical qubits.

Furthermore, it is clear that the (odd) rotated code with pure $Y$ noise is equivalent to the repetition code and, hence, has a threshold of $50\%$, in accordance with Corollary~\ref{cor:threshold}.

\section{\label{sec:performance-pure}Performance of surface codes with pure Y noise}

\begin{figure*}
  \includegraphics{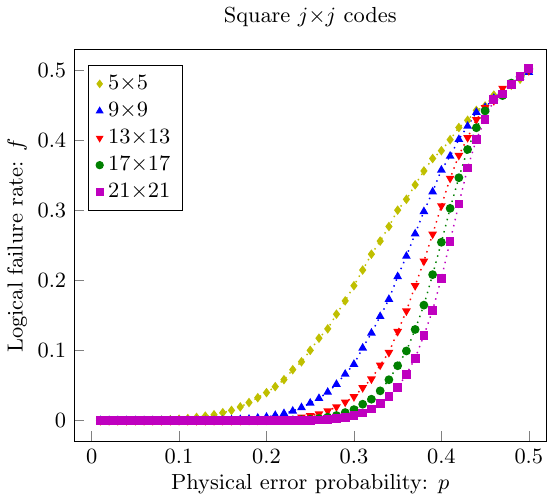}
  \includegraphics{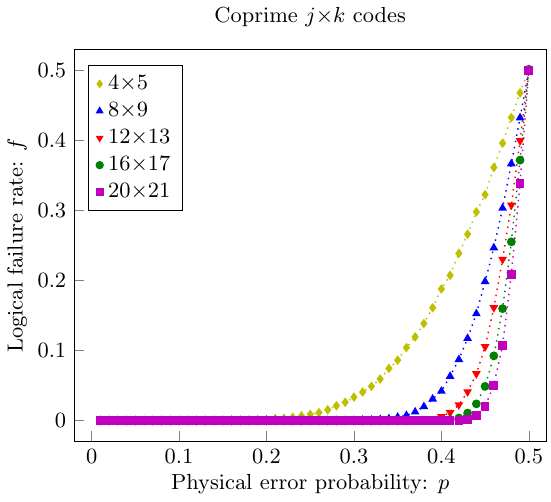}
  \includegraphics{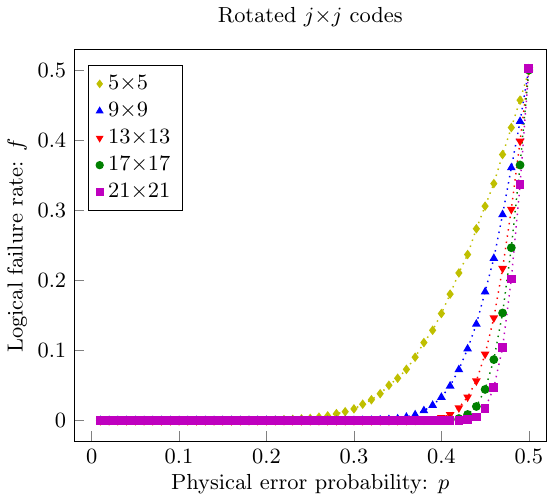}
  \caption{\label{fig:logical-failure-rate-v-physical-error-rate}
    Logical failure rate $f$ as a function of physical error probability $p$ for surface code families: square, coprime, and rotated, subject to pure $Y$ noise.
    Data points are sample means over 60\,000 runs using exact maximum-likelihood decoding.
    Dotted lines connect successive data points for a given code size.
  }
\end{figure*}
\begin{figure*}
  \includegraphics{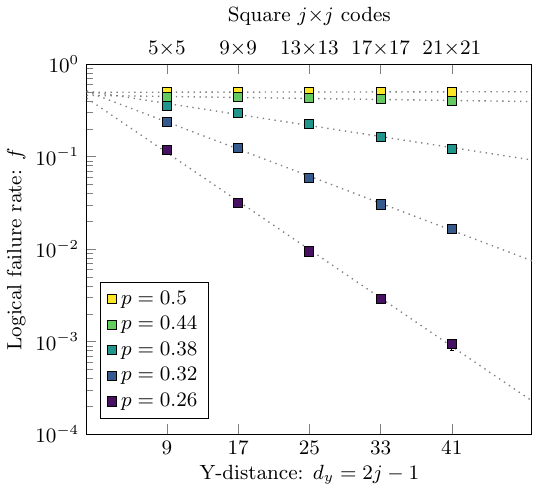}
  \includegraphics{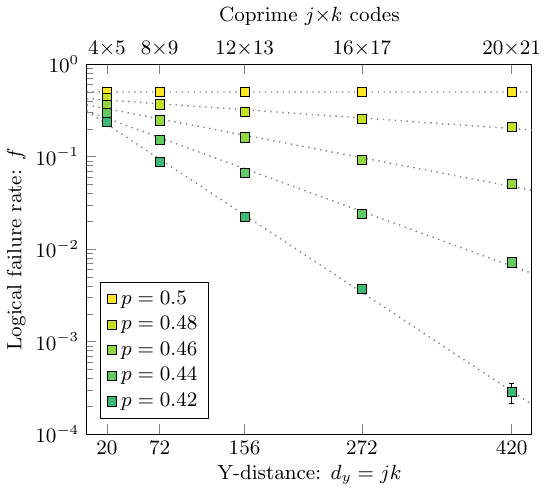}
  \includegraphics{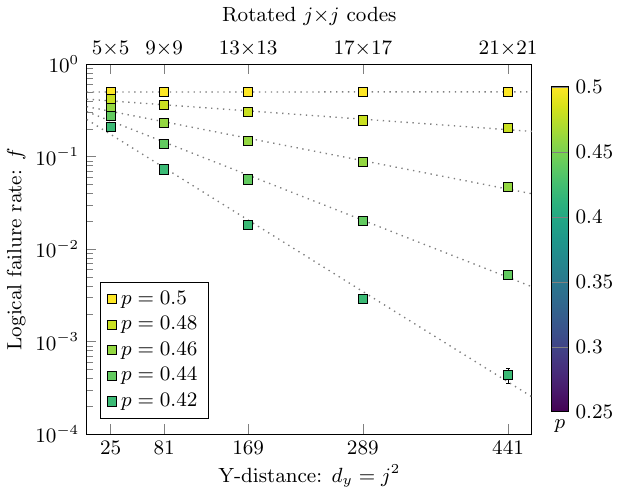}
  \caption{\label{fig:logical-failure-rate-v-y-code-distance}
    Exponential decay of the logical failure rate $f$ with respect to code distance $d_Y$ to pure $Y$ noise in the regime of physical error probability $p$ at and below the error threshold for surface code families: square, coprime, and rotated, subject to pure $Y$ noise.
    Data points are sample means over 60\,000 runs using exact maximum-likelihood decoding.
    Dotted lines indicate a least-squares fit to data for a given $p$, and error bars indicate one standard deviation.
  }
\end{figure*}

In Sec.~\ref{sec:y-noise-theorems}, we present our analytical results for surface codes with pure $Y$ noise, highlighting features that contribute to the ultrahigh threshold results with $Y$-biased noise, found in Ref.~\cite{Tuckett2018} and improved upon in Sec.~\ref{sec:thresholds-biased}.
Our analytical results also indicate that coprime and (odd) rotated codes should achieve lower logical failure rates than square codes with pure $Y$ noise.

Here, we present our numerical investigation into the performance of  surface codes with pure $Y$ noise.
In particular, we present results for square, coprime, and rotated surface code families, confirming the $50\%$ error threshold.
We also demonstrate a significant reduction in the logical failure rate for coprime and rotated codes compared with square codes.
Specifically, quadratically fewer physical qubits may be used to achieve a target logical failure rate by using coprime or rotated codes compared with square codes.

\subsection{Advantage of coprime and rotated surface codes with pure Y noise}
\label{sec:co-prime-advantage-pure}

We investigate the performance of surface codes with pure $Y$ noise.
Besides confirming the $50\%$ threshold for the surface code, we demonstrate a significant reduction in logical failure rate for coprime and (odd) rotated surface codes compared to square surface codes such that a target logical failure rate may be achieved with quadratically fewer physical qubits using coprime or rotated codes in place of square codes.
That is, we demonstrate that logical failure rate decays exponentially with $Y$-distance but since, in accordance with Corollary~\ref{cor:distance}, the $Y$-distance of square codes is $O(\sqrt{n})$ and that of coprime and rotated codes is $O(n)$, logical failure rate decays quadratically faster with $n$ for coprime and rotated codes, where $n$ is the number of physical qubits.

In Fig.~\ref{fig:logical-failure-rate-v-physical-error-rate}, we plot logical failure rate $f$ as a function of physical failure rate $p$ for surface codes belonging to the following families: square, coprime, and rotated codes.
For coprime and rotated codes, we see clear evidence of an error threshold at $p_c = 50\%$, consistent with Corollary~\ref{cor:threshold}.
For square codes, the data are consistent with a threshold of $p_c = 50\%$, but the evidence is less definitive.
Within a code family, it is expected that smaller codes will perform worse than larger codes below threshold.
However, comparing the performance of smaller coprime and rotated codes to square codes, we see a significant improvement in logical failure rate across the full range of physical error probabilities.
For example, the $21{\times}21$ rotated code, with $n=441$, and the $20{\times}21$ coprime code, with $n=800$, both clearly outperform the $21{\times}21$ square code, with $n=841$.
This result can be seen as a qualitative demonstration of the effect of the features of surface codes with pure $Y$ noise identified in Sec.~\ref{sec:y-noise-theorems}.

In Fig.~\ref{fig:logical-failure-rate-v-y-code-distance}, we plot logical failure rate $f$ as a function of code distance $d_Y$ to pure $Y$ noise at physical error probabilities $p$ at and below the threshold $p_c=50\%$ for surface codes belonging to the following families: square, coprime, and rotated codes.
For each code family, we see an exponential decay of the logical failure rate $f \sim \exp(-\alpha d_Y)$, where $\alpha$ is a function of $(p_c - p)$, which is consistent with the threshold $p_c = 50\%$ predicted by Corollary~\ref{cor:threshold}.
Considering $j{\times}k$ surface codes, according to Corollary~\ref{cor:distance}, $d_Y=2j-1$ for square codes, $d_Y=jk$ for coprime codes, and $d_Y=j^2$ for rotated codes.
That is, $d_Y=O(\sqrt{n})$ for square codes, and $d_Y=O(n)$ for coprime and rotated codes.
As a result, based on the observed exponential decay, quadratically fewer physical qubits are required to achieve a target logical failure rate for a given physical error rate by using coprime or rotated codes in place of square codes.

To investigate the performance of different families of surface codes with pure $Y$ noise, we sample the logical failure rate across a full range of physical error probabilities for square, coprime, and rotated codes.
For each code family, we use an exact maximum-likelihood decoder.
For square and coprime codes, we use the $Y$-decoder, defined in Appendix~\ref{sec:y-decoder}, to avoid the limitations of an approximate~\cite{Tuckett2018} or nonoptimal (see Sec.~\ref{sec:code-equivalence}) decoder.
For rotated codes, we use the tensor-network decoder, defined in Sec.~\ref{sec:tn-decoding-rotated-codes}, which is exact for pure $Y$ noise.
We use code sizes \{$5{\times}5$, $9{\times}9$, $13{\times}13$, $17{\times}17$, $21{\times}21$\} for square and rotated codes and \{$4{\times}5$, $8{\times}9$, $12{\times}13$, $16{\times}17$, $20{\times}21$\} for coprime codes and 60\,000 runs per code size and physical error probability.
In our decoder implementations, we use the Python language with SciPy and NumPy libraries~\cite{scipy,numpy} for fast linear algebra and the mathmp library~\cite{mpmath} for arbitrary-precision floating-point arithmetic.

\section{\label{sec:performance-biased}Performance of surface codes with biased noise}

Our analytical results (see Sec.~\ref{sec:y-noise-theorems}), highlight features of surface codes with pure $Y$ noise that contribute to ultrahigh thresholds with $Y$-biased noise (see Ref.~\cite{Tuckett2018}) and the improvement in logical failure rate achieved by coprime and rotated surface codes (see Sec.~\ref{sec:performance-pure}).

Here, we present our numerical investigation into the performance of surface codes with $Y$-biased noise.
In particular, we improve on the results of Ref.~\cite{Tuckett2018}, providing strong evidence that the threshold of the surface code tracks the hashing bound exactly for all biases.
We also demonstrate that the improvement in logical failure rate of coprime and rotated codes with pure $Y$ noise persists with $Y$-biased noise, such that a \emph{smaller} coprime or rotated code outperforms a square code for a wide range of biases.

\subsection{\label{sec:thresholds-biased}Thresholds of surface codes with biased noise}

In previous work~\cite{Tuckett2018}, we show that the surface code exhibits ultrahigh thresholds with $Y$-biased noise (equivalently, $Z$-biased noise on the modified surface code of Ref.~\cite{Tuckett2018}).
The results of Ref.~\cite{Tuckett2018} indicate that the threshold error rate of the surface code appears to follow the hashing bound for low to moderate bias; however, it is unclear whether the surface code saturates the hashing bound for all biases.

Here, we improve on the results of Ref.~\cite{Tuckett2018}, providing strong evidence that the threshold error rate of the surface code saturates the hashing bound exactly for all biases.
Our results are summarized in Fig.~\ref{fig:threshold-v-bias}, in which threshold estimates for a range of biases are plotted along with the hashing bound.
Error bars are one standard deviation relative to the fitting procedure. 
The threshold estimates are very close to the hashing bound, and any residual differences are likely due to finite size effects and decoder approximation.
We estimate the following thresholds: $18.8(2)\%$, $19.4(1)\%$, $22.3(1)\%$, $28.1(2)\%$, $33.9(2)\%$, $39.2(1)\%$, $42.9(2)\%$, and $45.4(2)\%$, with $\eta=0.5$ (standard depolarizing noise), 1, 3, 10, 30, 100, 300, and 1000, respectively.
The corresponding hashing bound values are $18.9\%$, $19.4\%$, $22.2\%$, $27.8\%$, $33.5\%$, $39.0\%$, $42.8\%$, and $45.6\%$, respectively.

These thresholds are all achieved using a particular tensor-network decoder.
The tensor-network decoder of Ref.~\cite{Bravyi2014}, used in Ref.~\cite{Tuckett2018}, is an approximate maximum-likelihood decoder tuned via a parameter $\chi$, allowing a trade-off between accuracy and computational cost.
In Ref.~\cite{Tuckett2018}, we use $\chi = 48$ to keep the simulations tractable, but we find the decoder is not completely converged in the intermediate- to high-bias regime.
Here, we improve on these results by using a tensor-network decoder, defined in Sec.~\ref{sec:tn-decoding-rotated-codes}, that adapts the decoder of Ref.~\cite{Bravyi2014} to rotated codes and achieves a much stronger convergence with biased noise.
The convergence of the decoder with bias is summarized in Fig.~\ref{fig:convergence-v-bias}, which shows an estimate of the logical failure rate for the rotated $33{\times}33$ surface code near threshold as a function of $\chi$ for a range of biases.
For each bias, the shift in logical failure rate, between the two largest $\chi$ shown, is less than half a standard deviation, assuming a binomial distribution.

\begin{figure}[t]
  \includegraphics{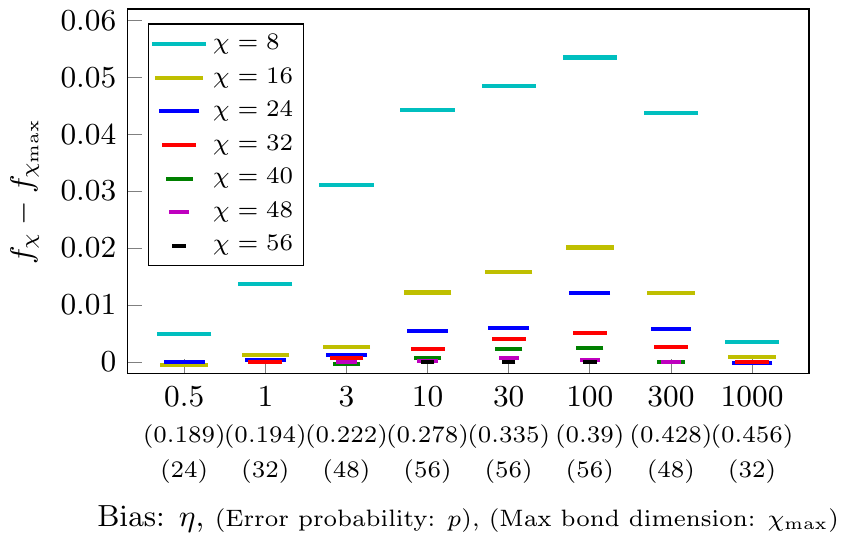}
  \caption{\label{fig:convergence-v-bias}
    Decoder convergence for the rotated $33{\times}33$ surface code, represented by shifted logical failure rate $f_\chi-f_{\chi_{\max}}$, as a function of $\chi$ at a physical error probability $p$ near the threshold for the given bias $\eta$ .
    Each data point corresponds to 30\,000 runs with identical errors generated across all $\chi$ for a given bias.
  }
\end{figure}

Our method to numerically estimate the threshold of the surface code with biased noise follows the general approach taken in Ref.~\cite{Tuckett2018}, with the key difference that we use the tensor-network decoder adapted to rotated codes (see Sec.~\ref{sec:tn-decoding-rotated-codes}) and choose $\chi$ such that the decoder more strongly converges.
We give a brief summary of the approach here and refer the reader to Ref.~\cite{Tuckett2018} for full details.
We use rotated surface codes of sizes $21{\times}21$, $25{\times}25$, $29{\times}29$, and $33{\times}33$.
We estimate the threshold for biases $\eta=0.5, 1, 3, 10, 30, 100, 300,$ and $1000$, where $\eta=p_Y/(p_X+p_Z)$ and $p_X=p_Z$ (see Sec.~\ref{sec:definitions}); we use decoder approximation parameter $\chi=16, 24, 40, 48, 48, 48, 40,$ and $24$, respectively, to achieve convergence to within less than half a standard deviation.
We run 30\,000 simulations per code size and physical error probability.
As in Ref.~\cite{Tuckett2018}, we use the critical exponent method of Ref.~\cite{Wang2003} to obtain threshold estimates with jackknife resampling over the code distances to determine error bounds.

\subsection{\label{sec:co-prime-advantage-biased}Advantage of coprime and rotated surface codes with biased noise}

In Sec.~\ref{sec:co-prime-advantage-pure}, we give a demonstration that coprime and rotated surface codes outperform square surface codes with pure $Y$ noise in terms of logical failure rate.
It is natural to ask if coprime and rotated codes also outperform square codes with $Y$-biased noise, i.e., when $X$ and $Z$ errors may also occur.
We demonstrate that a significant reduction in logical failure rate against biased noise can be achieved using a rotated $j{\times}j$ code, containing $n=j^2$ physical qubits, compared to a square $j{\times}j$ code, containing $n=2j^2-2j+1$ physical qubits.

Our results are summarized in Fig.~\ref{fig:coprime-v-square-small-code-comparison}, where we compare the rotated $9{\times}9$ code, containing 81 physical qubits, to the square $9{\times}9$ code, containing 145 physical qubits.
With standard depolarizing noise, i.e., $\eta=0.5$, and with a low bias, e.g.\ $\eta=10$ (where $Y$ errors are 10 times more likely than both $X$ and $Z$), we see similar performance for the rotated and square codes.
In the limit of pure $Y$ noise, we see the very large improvement, across the full range of physical error probabilities, that is already demonstrated in Sec.~\ref{sec:co-prime-advantage-pure}.
Most interestingly, the improvement remains large through the intermediate-bias regime, $\eta=100$, over a wide range of physical error probabilities, indicating that the advantage of rotated codes over square codes persists with modest noise biases.
We note that qualitatively similar results are observed when comparing the coprime $7{\times}8$ code to the square $8{\times}8$ code (not shown here).

The advantage of rotated codes with biased noise can be explained in terms of the features of surface codes with $Y$ noise identified in Sec.~\ref{sec:y-noise-theorems}.
The rotated $9{\times}9$ code has the same $X$- and $Z$-distance ($d_X=d_Z=9$) as the square $9{\times}9$ code.
However, the rotated code is much less sensitive to $Y$ noise, having a much larger $Y$-distance ($d_Y=81$) than the square code ($d_Y=17$) and having only one $Y$-type logical operator ($c_Y=1$) compared to many more such operators ($c_Y=2^8=256$) on the square code.
Therefore, for sufficient bias, we expect rotated $j{\times}j$ codes to outperform square $j{\times}j$ codes, despite containing approximately half the number of physical qubits.
Also, for a given bias, we expect the relative advantage to increase with code size, as the $Y$-sensitivity of the rotated code decreases faster than the $X$- or $Z$-sensitivity, until low-rate errors become the dominant source of logical failure, at which point high-rate errors are effectively suppressed.

To compare the performance of rotated and square codes with $Y$-biased noise, we
sample the logical failure rate across a full range of physical error probabilities for the square $9{\times}9$ code and the rotated $9{\times}9$ code with noise biases $\eta \in \{0.5, 10, 100, 1000, 10\,000, \infty\}$.
Sample means are taken over 30\,000 and 1\,200\,000 runs per bias and physical error probability for the square and rotated codes, respectively.
Since the noise is biased, we cannot use the $Y$-decoder (see Appendix~\ref{sec:y-decoder}) for exact maximum-likelihood decoding.  Instead, we use the tensor-network decoder of Ref.~\cite{Bravyi2014} for square codes and the tensor-network decoder of Sec.~\ref{sec:tn-decoding-rotated-codes} for rotated codes, both of which approximate maximum-likelihood decoding.
Both decoders are tuned via an approximation parameter $\chi$, which controls the trade-off between efficiency and convergence.
As explained in Sec.~\ref{sec:tn-decoding-rotated-codes}, the decoder adapted to rotated codes converges much more strongly with biased noise.
We choose $\chi=48$ for the square code decoder and $\chi=8$ for the rotated code decoder, such that both decoders converge, at the most challenging bias of $\eta=100$, to within less than half a standard deviation, relative to $\chi+8$, assuming a binomial distribution.
The use of relatively small codes ensures significant logical failure rates at low physical error probabilities and keeps computational requirements to a reasonable level.

\section{\label{sec:tn-decoding-rotated-codes}Improved tensor-network decoding of rotated codes with biased noise}

In this section, we describe how tensor-network decoding of the surface code under biased noise can be improved using the rotated surface code layout. We show that the rotated layout allows us to remove certain correlations present in the tensor network used for maximum-likelihood decoding~\cite{Bravyi2014}, allowing efficient and optimal decoding for pure $Y$ noise. The removal of such correlations greatly improves the efficiency of the decoder in the case of noise strongly biased toward $Y$, but with a small probability of $X$ and $Z$ errors, a situation previously shown to be challenging using the standard layout~\cite{Tuckett2018}.
Throughout this section, we refer to surface codes oriented as in Fig.~\ref{fig:rotated-code}(b), where shaded and blank faces correspond to $X$- and $Z$-type checks, respectively.

We briefly describe the approximate maximum-likelihood decoder proposed by Bravyi, Suchara, and Vargo in Ref.~\cite{Bravyi2014}.
Maximum-likelihood decoding for stochastic Pauli noise chooses the correction that has the highest probability of successfully correcting the error given an error syndrome, accounting for all errors consistent with that syndrome. If performed exactly it is, by definition, optimal.

The maximum-likelihood decoding algorithm in Ref.~\cite{Bravyi2014} is based on mapping coset probabilities to tensor-network contractions. The probability of a coset for an error $f$ is given by 
\begin{equation}
\label{e:coset_def}
\pi(f\mathcal{G})=\sum_{\alpha, \beta} T(\alpha; \beta), 
\end{equation}
where $T(\alpha; \beta)$ is defined as the probability of the Pauli error $f$ times the stabilizer $g(\alpha, \beta):=\prod_v (A_v)^{\alpha_v}\prod_p (B_p)^{\beta_p}$, where $\alpha_v, \beta_p\in \{0,1\}$ and the summation is over all bit strings $\alpha=\alpha_1, \alpha_2,\dots\alpha_{(n-1)/2}$ and $\beta=\beta_1, \beta_2,\dots\beta_{(n-1)/2}$. Because of the local structure of the surface code, this summation can be expressed as the contraction of a square-lattice tensor network. While there is some freedom in how the tensor network for a given coset can be defined on both the standard and rotated surface code layouts, we illustrate how a particular definition of the tensor network on the rotated surface code layout can result in significantly more efficient decoding of biased noise.

\begin{figure}
    \includegraphics[width=0.45\textwidth]{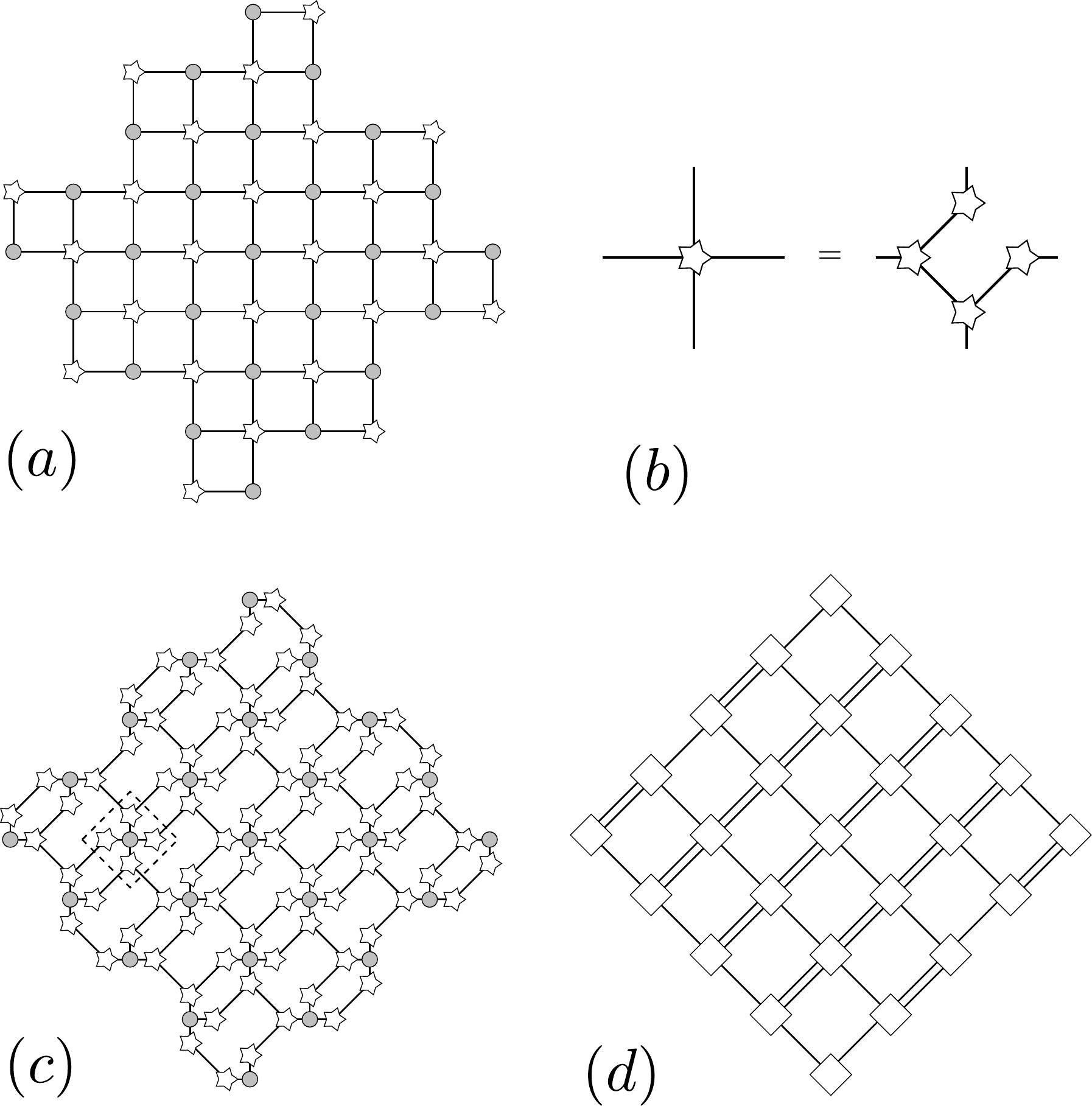}
    \caption{(a) Tensor network representing a coset probability for a rotated code. It consists of qubit tensors (circles) and check tensors (stars). The layout in (a) is obtained by applying the construction of Ref.~\cite{Bravyi2014} to the rotated code without modification. (b) Splitting a check tensor into multiple check tensors. This splitting is possible because check tensors take the value one if all indices are identical and zero otherwise. (c) A modified tensor network representing a coset probability where a single cell is outlined by a dashed box. This network is obtained from (a) by splitting check tensors. (d) Final modified tensor network obtained by contracting tensors in cells together to form merged tensors (squares). In the discussion of the contraction of this tensor network, we imagine rotating the network anticlockwise by $45^\circ$ and contracting from left to right. Note that this tensor network is not isotropic: In this rotated frame, the bond dimension is 2 for horizontal edges and is 4 for most vertical edges (except on the boundary).}
    \label{f:bevel_tn}
\end{figure}

A complete description of the tensor network that leads to more efficient decoding is provided in Fig.~\ref{f:bevel_tn}. We highlight the essential features that give rise to the improved decoding performance. In this layout, every tensor corresponds to a physical qubit, and a horizontal edge between columns $i$ and $i+1$ corresponds to a unique check that acts nontrivially on qubits in both of these columns. We illustrate the correspondence between checks and tensor-network edges on a $5\times5$ rotated code in Fig.~\ref{f:coordinates}.
\begin{figure}
  \includegraphics[width=0.4\textwidth]{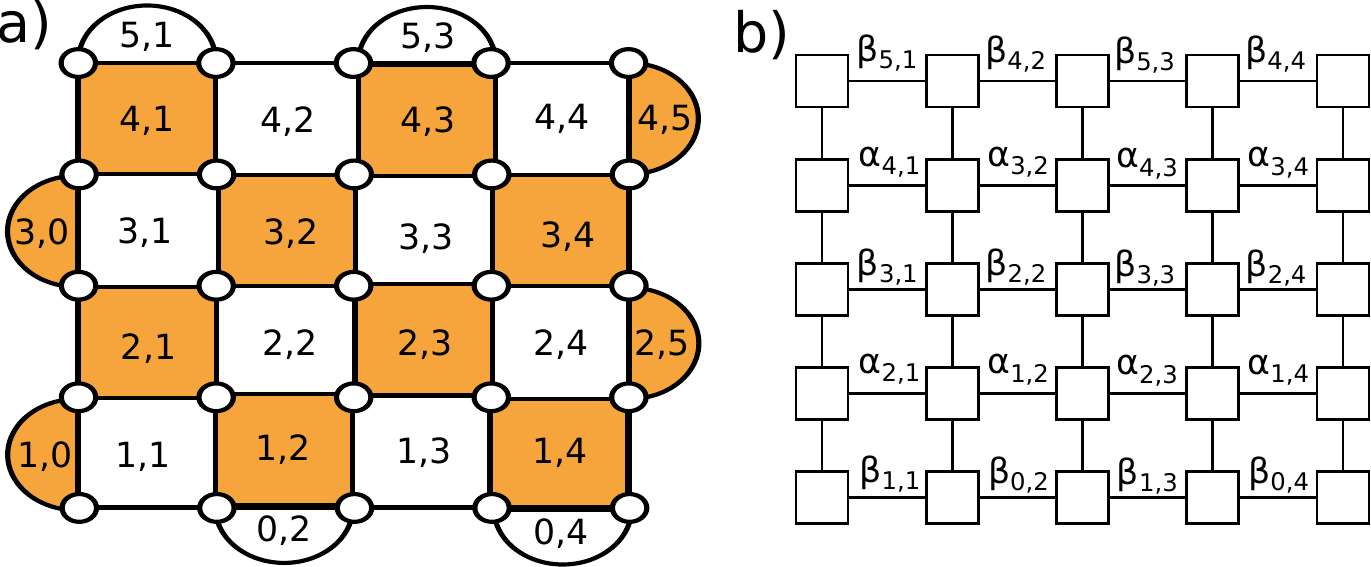}
  \caption{(a) Check coordinates are assigned to each check in the rotated layout. (b) The tensor network is defined such that each horizontal edge corresponds to a specific check. The $\alpha$ indices correspond to $X$ checks, and the $\beta$ indices correspond to $Z$ checks, with the subscripts indicating the check coordinates. Each tensor corresponds to a specific qubit. The bond dimension of horizontal edges is 2, while the bond dimension of the vertical edges is 4.}
  \label{f:coordinates}
\end{figure}

For certain structured instances of this problem, corresponding to independent $X$ or $Z$ flips, an efficient algorithm for contracting the network is known~\cite{Bravyi2014}. However, there is no known efficient algorithm for exact contraction of the network in the case of general local Pauli noise. 

In this case, an approximate method for evaluating the tensor-network contraction is used~\cite{Bravyi2014}. In this method, the leftmost boundary of the tensor network is treated as a matrix product state (MPS). Columns of the tensor network, which take the form of matrix product operators, are successively applied to the MPS until there are no columns left and the entire lattice is contracted. 

An approximation is used to keep this calculation tractable. After each column is applied, the singular value decomposition is used to reduce the size of the tensors in the MPS, effectively keeping only the $\chi$ largest Schmidt values for each bipartition of the chain and setting the remainder to zero. Without such a truncation, the number of parameters describing the tensors increases exponentially in the number of columns applied to the MPS. The parameter $\chi$ can be controlled independently, with larger $\chi$ improving accuracy at an increased computational cost. The overall runtime of the algorithm is $\mathcal{O}(n\chi^3)$. 

Surprisingly, on the rotated code with the tensor network described above, there is no entanglement in the boundary MPS for pure $Y$ noise. In other words, the MPS decoder is exact for $\chi=1$, independent of system size. This result is in contrast to the standard layout, where $\chi\sim 48$ is required for a reasonable approximation to coset probabilities on a $21\times21$ system~\cite{Tuckett2018}.

As we explain in the following section, this improvement can be attributed to the boundary conditions of the code and the layout of the tensor network. While exact decoding of $Y$ noise can also be performed using methods described in Appendix~\ref{sec:y-decoder}, the MPS decoder can be extended easily to noise that is mostly $Y$ noise but with nonzero probability of $X$ and $Z$ errors.  Our convergence results (see Sec.~\ref{sec:thresholds-biased}) show that there is a substantial improvement in the performance over the standard method when applied to this type of noise.

We observe a similar improvement in performance using the tensor-network decoder described in Ref.~\cite{Darmawan2018} when defined on the rotated layout and with an analogous tensor-network layout. Exact decoding is achieved with $\chi=4$ for pure $Y$ noise, which is not as efficient as the improved MPS decoder described above but substantially more efficient than the MPS decoder on the standard layout.

We remark that, on the standard layout, changing from a square lattice to coprime does little to improve the performance of the MPS decoder. Since the contraction algorithm proceeds column by column, a $21\times21$ (square) code and a $21\times22$ (coprime) code have an identical boundary MPS after the first $20$ columns are contracted if the same error is applied to qubits in these columns. Thus, we expect the error resulting from the truncation of small singular values during the contraction to be at least as bad for the $21\times22$ code as the $21\times21$ code.

\subsection{Boundary entanglement in MPS decoder}
We show that the boundary MPS of the rotated code with the above tensor-network layout is unentangled in the case of infinite bias. The boundary MPS appearing in the contraction algorithm is a (generally approximate) representation of the ``boundary state", obtained by contracting all indices of the network up to a given column and leaving the right-going indices of that column uncontracted. More precisely, we define $\psi(\alpha^R; \beta^R)$ to be the contraction of the network up to the $j<L$ column, with the right-going boundary indices set to $\alpha^R; \beta^R$. The $L$-qubit boundary state is defined as $\ket{\psi^R}:=\sum_{\alpha^R; \beta^R}\psi(\alpha^R; \beta^R)\ket{\alpha^R; \beta^R}$. We illustrate such a boundary state in Fig.~\ref{f:check_error_config_tn}(a). Let $Q_j$ be the set of qubits in columns up to and including column $j$. As previously described, each boundary index in $\alpha^R=\alpha_{1,j}, \alpha_{3,j}, \dots, \alpha_{L-2,j}$ and $\beta^R=\beta_{0,j}, \beta_{2,j}, \dots, \beta_{L-1,j}$ corresponds to a check acting nontrivially on qubits in columns $j$ and $j+1$, where the check subscripts here are for odd $j$ (for even $j$, simply add 1 to every row index).

We call checks that act only nontrivially on qubits contained in $Q_j$ bulk checks and refer to them using the indices $\alpha^B; \beta^B$, with superscript $B$. We refer to a specific $\alpha^R; \beta^R$ as a boundary configuration and a specific $\alpha^B; \beta^B$ as a bulk configuration. We define $\mathcal{G}'\subseteq \mathcal{G}$ to be the set of stabilizer elements that act nontrivially only on $Q_j$ and $g(\alpha^B; \beta^B)\in \mathcal{G}'$ to be the stabilizer element corresponding to the bulk configuration $\alpha^B; \beta^B$. We define $h(\alpha^R; \beta^R)$ to be the product of boundary checks corresponding to the boundary configuration $\alpha^R; \beta^R$ whose action is restricted to qubits in $Q_j$ (so the action on qubits in column $j+1$ is ignored).
The fact that $\psi(\alpha^R; \beta^R)$ is calculated by contracting all indices to the left of column $j$ means that all bulk configurations $\alpha^B; \beta^B$ are summed over, like in Eq.~\eqref{e:coset_def} but restricted to checks in the first $j$ columns. So we can write
\begin{equation}
\label{e:coset_reduced}
\psi(\alpha^R; \beta^R)=\pi'(f' \mathcal{G}') = \sum_{\alpha^B, \beta^B} T'(\alpha^B; \beta^B;\alpha^R; \beta^R), 
\end{equation}
where $\pi'$, $f'$, and $T'$, respectively, correspond to versions of $\pi$, $f$, and $T$ that are restricted to $Q_j$. Specifically, $f'$ is the Pauli error $f$ restricted to $Q_j$, and $T'(\alpha^B; \beta^B;\alpha^R; \beta^R)$ is the probability of the Pauli error $f'g(\alpha^B;\beta^B)h(\alpha^R;\beta^R)$ on qubits in $Q_j$. We illustrate an example error $f'$, bulk configuration $\alpha^B; \beta^B$, and boundary configuration $\alpha^R; \beta^R$ in Fig.~\ref{f:check_error_config_tn}. The coset probability $\pi'$ is likewise restricted to qubits $Q_j$, with the boundary checks fixed.

\begin{figure}
  \includegraphics[width=0.45\textwidth]{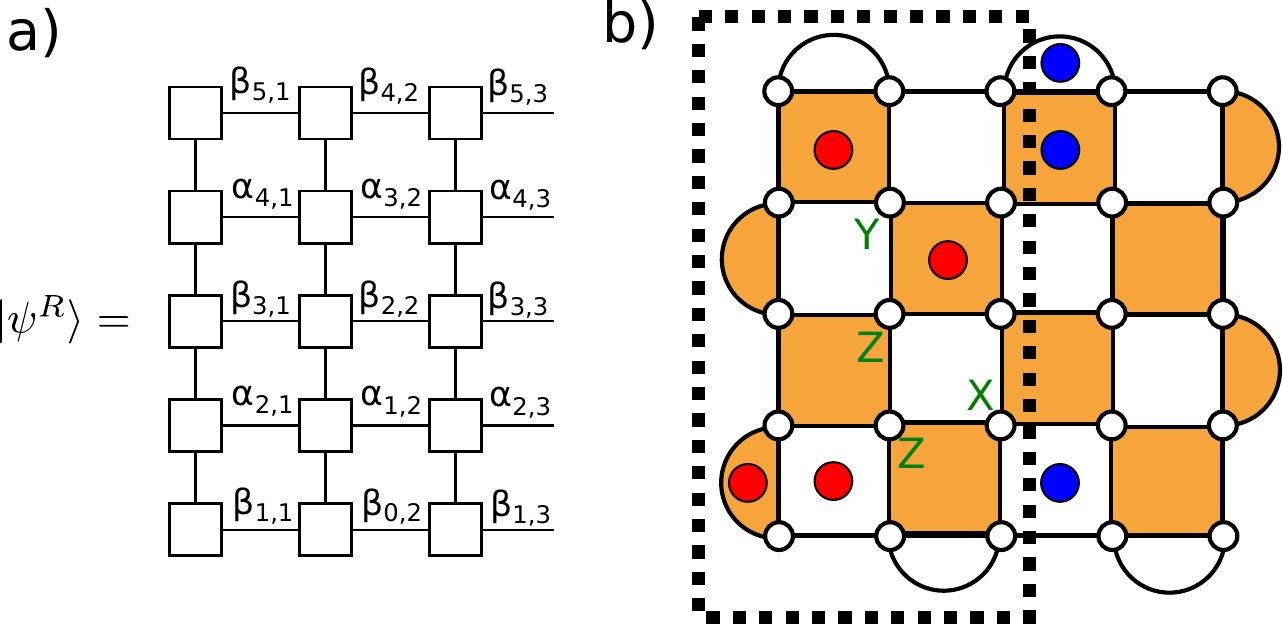}
  \caption{(a) A boundary state obtained by contracting the network up to a given column, and leaving the right-going indices of that column uncontracted. (b) An example check and error configuration illustrated for calculating the boundary state for the third column $j=3$ of a $5\times5$ code with rotated layout. A bulk configuration $\alpha^B;\beta^B$ is represented by the red dots, and a boundary configuration $\alpha^R;\beta^R$ is represented by the blue dots (where a dot on a check indicates that the check variable $\alpha_{i,k}$ or $\beta_{i,k}$ is 1). An error $f'$ is represented by green letters. Note that for this calculation we consider only the action of the checks and error on qubits in the first three columns $Q_3$, inside the dashed box. So the action on the boundary checks on qubits outside the box is ignored. The quantity $T'(\alpha^B;\beta^B;\alpha^R;\beta^R)$ is the probability of the product of all dotted checks and the error $f'$ in the dashed box. In the example configuration depicted above, this product contains four $X$, six $Y$ and two $Z$ errors, giving $T'(\alpha^B;\beta^B;\alpha^R;\beta^R)=p_X^4p_Y^6p_Z^2p_I^3$. In order to calculate the boundary state, all possible configurations of bulk checks must be summed over.  In the special case of pure $Y$ noise, where $p_X=p_Z=0$, only at most one term in this sum is nonzero for any boundary configuration. }
  \label{f:check_error_config_tn}
\end{figure}

In the case of pure $Y$ noise, the summation on the right-hand side of Eq.~\eqref{e:coset_reduced} simplifies dramatically. In fact, for any given choice of boundary variables $\alpha^R; \beta^R$ and error $f'$, there is at most one choice of $\alpha^B$ and $\beta^B$ such that $T'(\alpha^B; \beta^B;\alpha^R; \beta^R)$ is nonzero. So, given $\alpha^R; \beta^R$ and $f'$, either $\psi(\alpha^R; \beta^R)$ is zero or there exists a unique $\alpha^B, \beta^B$ such that 
\begin{equation}
\psi(\alpha^R; \beta^R)= T'(\alpha^B; \beta^B;\alpha^R; \beta^R).
\end{equation}
For a given $f'$ and $\alpha^R, \beta^R$, we say that a qubit is ``satisfied" for a given check configuration $\alpha^B; \beta^B$ if $f'g(\alpha^B; \beta^B)h(\alpha^R; \beta^R)$ acts on every qubit as either $I$ or $Y$ and not $X$ or $Z$. For pure $Y$ noise,  in order for $T'(\alpha^B; \beta^B, \alpha^R; \beta^R)$ to be nonzero, all qubits in $Q_j$ must be satisfied. We can solve for a bulk configuration $\alpha^B; \beta^B$ that satisfies all qubits, if one exists, by fixing check variables to satisfy qubits one at a time, starting from the qubit adjacent to the two-qubit boundary check in column $j$. There is only one bulk check adjacent to this qubit; therefore, only one choice for the corresponding check variable will satisfy that qubit. This fixes the first bulk check. 
We then proceed down this column to fix every check variable in the same manner.  With the check configuration in column $j$ determined, we then solve for checks in columns $j-1$, $j-2$, etc., in the same way until all check variables are determined, thereby solving for the bulk configuration $\alpha^B; \beta^B$. 

Note that, for certain $f'$ and $\alpha^R; \beta^R$, there may be no configuration of bulk checks that will satisfy all qubits, which implies that the $f'$ and $\alpha^R; \beta^R$ are not compatible with pure $Y$ noise, i.e., $\psi(\alpha^R, \beta^R)=0$. In fact, only a few special boundary configurations are compatible with pure $Y$ noise. We describe the boundary configurations $\alpha^R; \beta^R$ that are compatible with a given $f$, starting with the case of the trivial coset $f=I$. We show that the allowed bulk and boundary configurations consist of horizontal strings which terminate at two-qubit $X$ checks on the left code boundary, as shown in Fig.~\ref{f:bulk_boundary}. Other cosets (with $f\ne I$) follow straightforwardly from this.

\begin{figure}
  \includegraphics[width=0.5\textwidth]{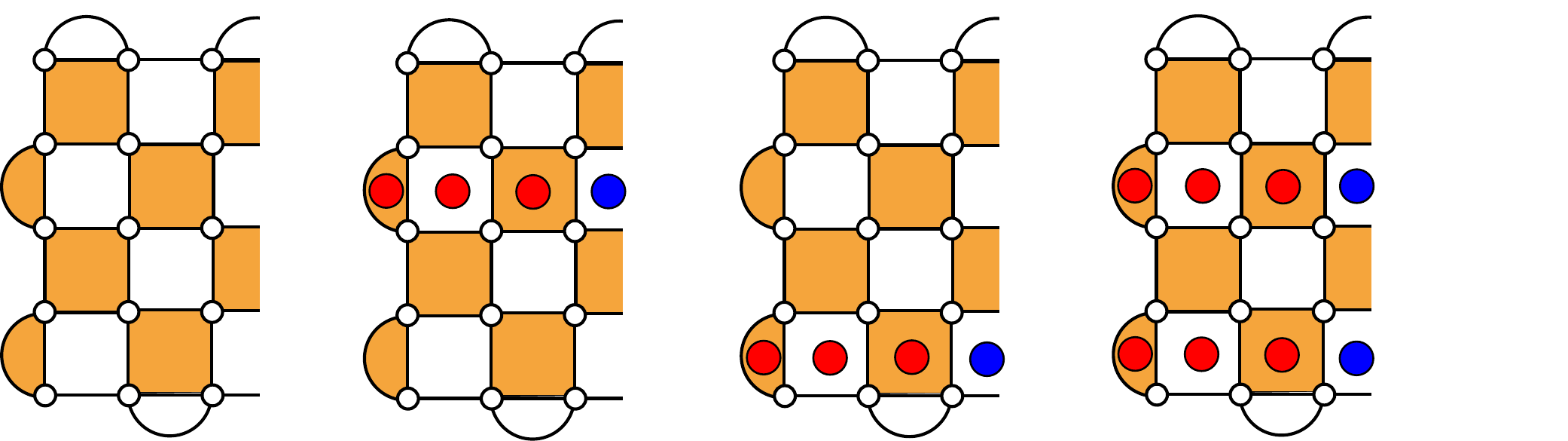}
  \caption{All allowed bulk and boundary configurations for $Y$ noise illustrated for the boundary state for the third column $j=3$ of a $5\times 5$ code on the rotated layout for the trivial coset $f'=I$. The product of the dotted checks must result in only $I$ and $Y$ errors on $Q_3$ (and no $X$ or $Z$ errors). Blue dots represent the boundary configuration, while red dots represent the corresponding bulk configuration. Blue dots must be connected to a two-qubit check on the left boundary by a string of red dots. The fact that these strings never overlap and are absorbed at the left boundary implies that the boundary variables are uncorrelated, and, therefore, there is no entanglement in the boundary state.}
  \label{f:bulk_boundary}
\end{figure}

We start from the left-hand side of the code and try to find bulk configurations that satisfy all qubits. We work our way up the column, finding relations between checks. We use the convention that qubit $(i, k)$ refers to the qubit on the bottom-left vertex of the check (face) with coordinates $(i,k)$, as in Fig.~\ref{f:coordinates}. First, in order to satisfy qubit $(1,1)$, we require $\alpha_{1, 0}=\beta_{1,1}$. With the parity of these checks fixed, in order to satisfy qubit $(2,1)$, we need $\alpha_{1,1}=0$. Then, to satisfy qubit $(3,1)$ we require $\alpha_{3, 0}=\beta_{3,1}$. Continuing up the column, we see that $\alpha_{i,0}=\beta_{i,1}$ for $i$ odd and $\alpha_{i, 1}=0$ for $i$ even, and the two-qubit $Z$ check at the end of the column satisfies $\beta_{(L+1)/2, 1}=0$. We can then solve for checks in the next column, finding $\alpha_{i,2}=\beta_{i,1}$ for odd $i$, $\beta_{i,2}=0$ for even $i$, and $\beta_{2, 0}=0$ for the two-qubit check on the lower boundary. 

Proceeding in this manner, we can solve for all the checks up to any given column. For the trivial coset, the bulk and boundary configurations satisfy the following:
\begin{itemize}
  \item All check variables in a given row must be take the same value.
  \item Check variables in rows terminated by a two-qubit $X$ check (odd rows) may take values 0 or 1. The remaining checks must take the value 0. 
\end{itemize}
We can easily calculate the probability of each satisfying check configuration. First, the trivial configuration $\alpha^R; \beta^R=0;0$ corresponding to the bulk configuration $\alpha^B; \beta^B=0;0$, i.e., with all bulk and boundary check variables set to $0$, has probability $(p_I)^{jL}$. Flipping any odd boundary check flips the corresponding row of checks in the bulk and introduces $2j$ $Y$ errors, changing the probability by a factor of $(p_Y/p_I)^{2j}$.  The fact that the weight introduced by flipping any row does not depend on which other rows are flipped implies that the boundary variables are independent and $\ket{\psi^R}$ is a product state which can be explicitly written as 
\begin{equation}
\ket{\psi^R}=\ket{0}_{\rm{end}}\bigotimes_{k\,\rm{even}}\ket{0}_k \bigotimes_{l \,\rm{odd}} \ket{\theta}_l,
\end{equation}
where $\ket{\theta}=p_I^{2j}\ket{0}+p_Y^{2j}\ket{1}$ and $\ket{0}_{\rm{end}}$ corresponds to the two-qubit $Z$-check at the end of the column. Since the boundary state for the trivial coset $f'=I$ is completely unentangled, the tensor network corresponding to this coset can be contracted exactly with $\chi=1$.

The case of a nontrivial coset with $f'\neq I$ is analogous to the case of a trivial coset. Starting from any satisfying bulk and boundary configuration, we  obtain all other satisfying bulk and boundary configurations by flipping odd rows of checks, as in the trivial coset. If we assume there exist satisfying bulk and boundary configurations ${\alpha^R}'; {\beta^R}'$ and ${\alpha^B}'; {\beta^B}'$, respectively, for a given error $f'$, the boundary state can be explicitly written as 
\begin{equation}
\ket{\psi^R}=\ket{{\beta^R}_{\rm{end}}'}\bigotimes_{k\,\rm{even}}\ket{{\gamma^R}_k'} \bigotimes_{l \,\rm{odd}} \ket{\theta(l)}_l,
\end{equation}
where $\ket{\theta(l)}=p_Y^{N(l)} p_I^{2j- N(l)}\ket{0}+p_Y^{2j- N(l)} p_I^{N(l)} \ket{1}$, $N(l)$ is the number of qubits on which $Y$ is applied in the rows adjacent to the $l$th row of checks when the boundary variable for row $l$ is 0 and where ${\gamma^R}'={\alpha^R}'$ for odd $j$ and ${\gamma^R}'={\beta^R}'$ for even $j$, and ${\beta^R_{\rm end}}'$ corresponds to the two-qubit check at the end of the column.

Therefore, using the tensor-network layout described above, any coset can be calculated exactly using the MPS decoder with $\chi=1$, which is a particular property of the physical boundary conditions of the code. In this case described above, starting from a vacuum (with all checks unflipped), flipping a boundary check results in a line of checks being flipped through the bulk, which is absorbed by a two-qubit check on the boundary. We call such a line of flipped check variables a ``lineon".

While we can define the tensor network analogously for the standard surface code layout, the boundary state does not have the same product state form. We find that the three-qubit boundary checks result in long-range correlations in the boundary state, which is because the three-qubit checks reflect lineons rather than absorb them, as illustrated in Fig.~\ref{f:standard_layout}. 
\begin{figure}
  \includegraphics[width=0.5\textwidth]{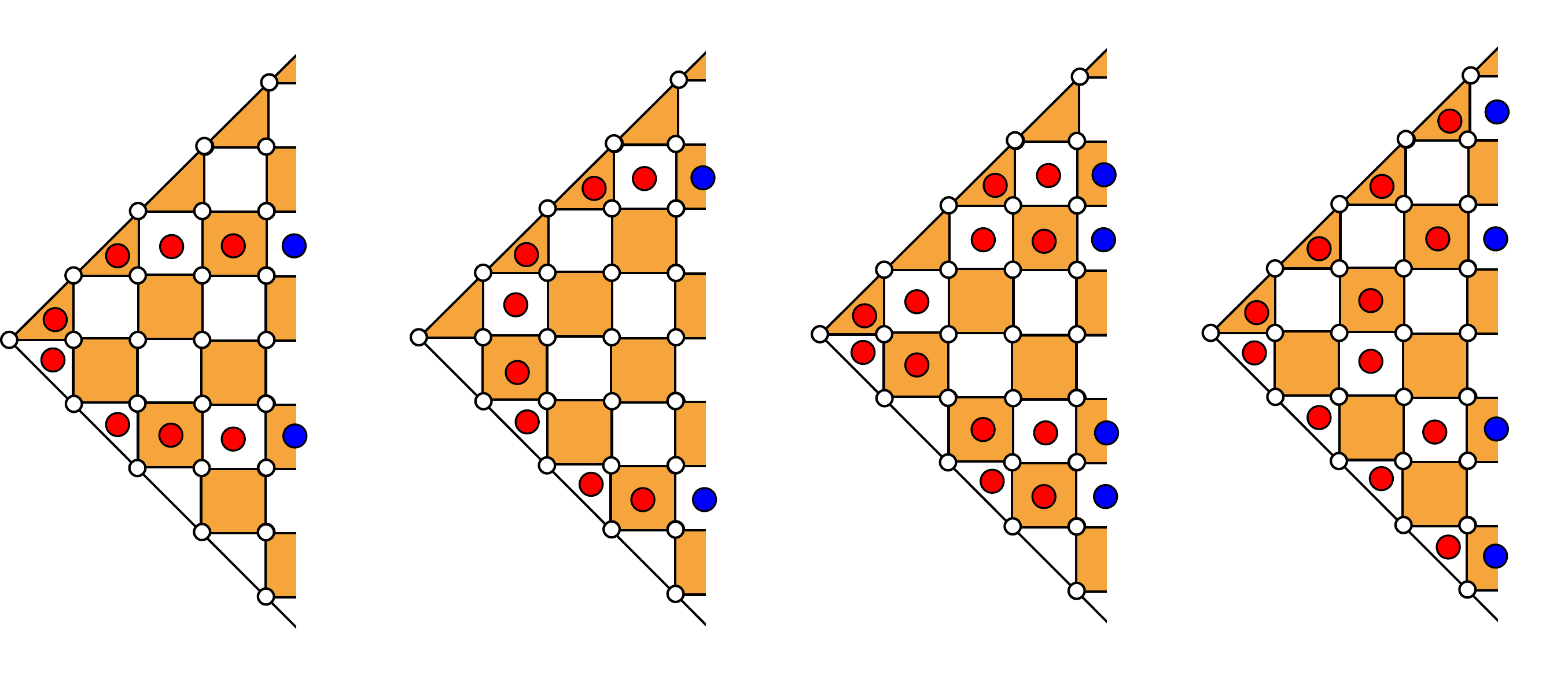}
  \caption{Some examples of boundary and bulk configurations for the standard layout of the surface code with three-qubit checks on the boundary for the trivial coset $f'=I$. The lineons travel in straight lines through the four-qubit bulk checks, but the three-qubit boundary checks have the effect of reflecting them by $\ang{90}$, such that they emerge on the right boundary on the exact opposite side. Therefore, each boundary variable is perfectly correlated with the boundary variable on the exact opposite side. Separate lineons can also cross paths, resulting in a cancellation of the bulk variables. Also, for neighboring pairs of lineons, $Y$ on the qubits shared between them cancel. These all result in correlations between the boundary variables and, therefore, entanglement in the boundary MPS.}
  \label{f:standard_layout}
\end{figure}
This result means that separated pairs of boundary checks must be flipped together. Also, when distinct lineons travel next to each other or cross, there is a cancellation of $Y$ errors. The consequence of this cancellation is that the probability of a particular lineon depends on whether other lineons are present, which results in correlations between boundary variables and entanglement in the boundary state. The rotated layout with two-qubit checks does not suffer from these problems. The lineons never cross; they are always separated by a row and are absorbed at the boundary.

To summarize this section, we show that the MPS decoder adapted to the rotated layout is exact with $\chi=1$ for pure $Y$ noise. This result is due to the fact that many correlations in the tensor network are eliminated in this case, making contraction of the tensor network much more efficient. This decoder can also take into account finite bias (i.e., nonzero $p_X$ and $p_Z$), and the improvement in efficiency also carries over to this case, as the numerical results of Sec.~\ref{sec:thresholds-biased} show.

\section{\label{sec:discussion}Discussion}

In this paper, we describe the structure of the surface code with pure $Y$ noise and show that this structure implies a $50\%$ error threshold and a significant performance advantage in terms of logical failure rate with coprime and rotated codes compared to square codes.
Furthermore, we provide numerics confirming our analytical results with pure $Y$ noise and demonstrating the performance advantage of rotated codes with $Y$-biased noise.
It is important to note that our results apply equally to pure $Z$ noise, i.e., dephasing noise, and the $Z$-biased noise prevalent in many quantum architectures, through the simple modification~\cite{Tuckett2018} of the surface code that exchanges the roles of $Z$ and $Y$ operators in stabilizer and logical operator definitions.
We, therefore, identify and characterize the features of surface codes that contribute to their ultrahigh thresholds with $Z$-biased noise and to the improvements in logical failure rate with coprime and rotated codes demonstrated in this paper.

In the limit of pure $Y$ noise, we show that the standard surface code is equivalent to a concatenation of classical codes: a single top-level cycle code and a number of bottom-level repetition codes.
We show that this implies the surface code with pure $Y$ noise has a threshold of $50\%$ and, for $j{\times}k$ surface codes with small $g = \gcd(j,k)$, the more effective repetition code dominates, leading to a reduction in logical failure rate.
In terms of logical operators, we show that $Y$-type logical operators are rarer and heavier than $X$- or $Z$-type equivalents, and coprime codes, in particular, have only one $Y$-type logical operator, and its weight is $O(n)$.
We also show that rotated codes, with odd linear dimensions, are closely related to coprime codes, admitting a single $Y$-type logical operator of optimal weight $n$.

We confirm, numerically, the $50\%$ error threshold of the surface code with pure $Y$ noise and demonstrate that coprime and rotated codes with pure $Y$ noise significantly outperform similar-sized square codes in terms of logical failure rates such that a target logical failure rate may be achieved with quadratically fewer physical qubits using coprime and rotated codes.
Furthermore, we demonstrate that this advantage persists with $Y$-biased noise.
In particular, we find that a \emph{smaller} rotated code, with approximately half the number of physical qubits, outperforms a square code, over a wide range of physical error probabilities, for biases as low as $\eta=100$, where $Y$ errors are 100 times more likely that $X$ or $Z$ errors.
We argue that, for a given bias, the relative advantage of coprime and rotated codes over square codes increases with code size, until low-rate errors become the dominant source of logical errors and high-rate errors are effectively suppressed.

Leveraging features of the structure of rotated codes with pure $Y$ noise, we define a tensor-network decoder that achieves exact maximum-likelihood decoding with pure $Y$ noise and converges much more strongly with $Y$-biased noise than the  decoder of Ref.~\cite{Bravyi2014}, from which it is adapted.
With this decoder, we are able to improve upon the results of Ref.~\cite{Tuckett2018} and provide strong evidence that the threshold error rate of surface codes tracks the hashing bound exactly for all biases, addressing an open question from Ref.~\cite{Tuckett2018}.
Saturating this bound is a remarkable result for a practical topological code limited to local stabilizers.

Although our analytical results focus on features of the surface code with pure $Y$ noise, it is interesting to put our observations of the performance of surface codes with biased noise in the context of other proposals to adapt quantum codes to biased noise~\cite{Ioffe2007, Sarvepalli2009, LaGuardia2014, Robertson2017, Aliferis2008, Aliferis2009, Stephens2013, Stephens2008, Napp2013, Brooks2013, Li2018, Rothlisberger2012, Xu2018}.
Several proposals have been made for constructing asymmetric quantum codes for biased noise from classical codes~\cite{Ioffe2007, Sarvepalli2009, LaGuardia2014, Robertson2017} (see Ref.~\cite{LaGuardia2014} for an extensive list of references), but of particular interest here are approaches that can be applied to topological codes.
A significant increase in threshold with biased noise has been demonstrated by concatenating repetition codes at the bottom level with another, possibly topological, code at the top level~\cite{Aliferis2008, Aliferis2009, Stephens2013}; interestingly, this construction mirrors the structure we find to be inherent to the surface code.
Performance improvements with biased noise have also been demonstrated by modifying the size and shape of stabilizers in Bacon-Shor codes~\cite{Stephens2008, Napp2013, Brooks2013} and surface and compass codes~\cite{Li2018}, by randomizing the lattice of the toric code~\cite{Rothlisberger2012} or by concatenating a small $Z$-error detection code to the surface code~\cite{Xu2018}. These approaches are distinct from the use of coprime or rotated codes (with the modification of Ref.~\cite{Tuckett2018}), which maintain the size and locality of surface code stabilizer generators, and so they could potentially be combined to yield further performance improvements.

Looking forward, the identified features of surface codes and the insights behind them suggest several interesting avenues of research.
For the surface code, specifically, different geometries may be more robust to logical errors than coprime and rotated codes in the high-bias regime, where a few well-placed $X$ and $Z$ errors can combine with strings of $Y$ errors to produce more common, lower-weight logical operators.
Similarly, certain geometries of surface code used to encode multiple qubits~\cite{Delfosse2016} may or may not maintain the high performance of simple surface codes with biased noise.
For topological codes, more generally, one can ask which codes exhibit an increase in performance with biased noise and what are the family traits of such codes; we have seen, for example, that the standard triangular 6.6.6 color code does not exhibit an increase in performance.
(Although this color code is equivalent, in some sense, to a folded surface code~\cite{Kubica2015}, the mapping that relates the two does not preserve the biased noise model.)

Finally, although this paper focuses on features of surface codes with $Y$ or $Y$-biased noise rather than the issue of fault-tolerant decoding, our numerical results motivate the search for fast fault-tolerant decoders for the surface code with biased noise.
The highly significant question of whether the high performance of surface codes with biased noise can be preserved in the context of fault-tolerant quantum computing, is addressed in a forthcoming paper~\cite{Tuckett2019}, where a fast but suboptimal decoder for tailored surface codes achieves fault-tolerant thresholds in excess of $5\%$ with biased noise.
Investigating the optimal fault-tolerant thresholds with biased noise and the performance well below threshold remain important avenues of research.

\vspace*{1em}
\begin{acknowledgments}

This work was supported by the Australian Research Council (ARC) via the Centre of Excellence for Engineered Quantum Systems (EQUS) Project No. CE170100009 and Discovery Project No. DP170103073.
S.B. acknowledges support from the IBM Research Frontiers Institute. 
A.S.D. was supported by JST, PRESTO Grant No. JPMJPR1917, Japan.
Access to high-performance computing resources was provided by the National Computational Infrastructure (NCI), which is supported by the Australian Government, and by the Sydney Informatics Hub, which is funded by the University of Sydney.
Some of the numerical computation in this work was carried out at the Yukawa Institute Computer Facility.

\end{acknowledgments}

\appendix

\section{\label{sec:color-thresholds}Color-code thresholds with biased noise}

We demonstrate that the threshold of the triangular 6.6.6 color code~\cite{Bombin2006} decreases when the noise is biased.
This result is in stark contrast to the surface code, which exhibits a significant increase in threshold with biased noise~\cite{Tuckett2018}.
Our results are summarized in Fig.~\ref{fig:color-threshold-v-bias}, in which, we contrast our results for the color code with those for the surface code, reproduced from Sec.~\ref{sec:thresholds-biased}.
From statistical physics arguments, the optimal error threshold for the unmodified surface code with pure $Z$ noise is estimated to be $10.93(2)\%$~\cite{Dennis2002,Merz2002}, and with depolarizing noise it is estimated to be $18.9(3)\%$~\cite{Bombin2012}.
The color code has similar error thresholds \cite{Katzgraber2009,Bombin2012} to the surface code with pure $Z$ noise and depolarizing noise.
Our results for the color code, using an approximate maximum-likelihood decoder, reveal a decrease in threshold with $Y$-biased noise: $18.7(1)\%$ with standard ($\eta=0.5$) depolarizing noise, $13.3(1)\%$ with bias $\eta=3$, $11.4(2)\%$ with bias $\eta=10$, $10.6(2)\%$ with bias $\eta=100$, and $10.5(2)\%$ in the limit of pure $Y$ noise.
In contrast, our results for the surface code, from Sec.~\ref{sec:thresholds-biased}, reveal a significant increase in threshold with $Y$-biased noise: $18.8(2)\%$ with standard ($\eta=0.5$) depolarizing noise, $22.3(1)\%$ with bias $\eta=3$, $28.1(2)\%$ with bias $\eta=10$, $39.2(1)\%$ with bias $\eta=100$, and the analytically proven $50\%$ threshold in the limit of pure $Y$ noise; see Sec.~\ref{sec:y-threshold}.
Our decoder implementation and numerics are described below.
The features of surface codes that contribute to their exceptional performance with biased noise are discussed in the body of the paper.

\begin{figure}[ht]
  \includegraphics{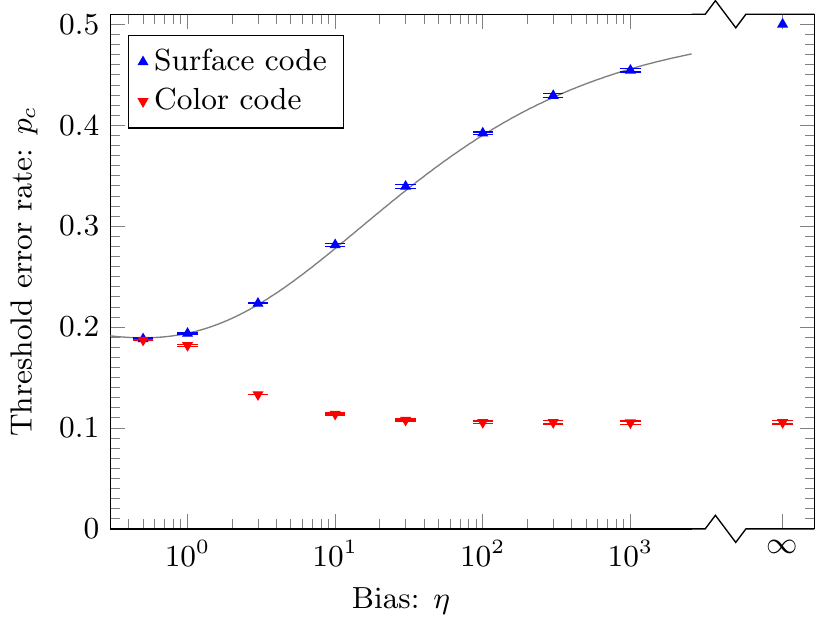}
  \caption{\label{fig:color-threshold-v-bias}
    Threshold error rate $p_c$ as a function of bias $\eta$.
    Red inverted triangles show threshold estimates for the triangular 6.6.6 color code.
    For comparison, blue triangles show threshold estimates for the surface code (reproduced from Sec.~\ref{sec:thresholds-biased}), with the point at infinite bias, i.e., only $Y$ errors, indicating the analytically proven $50\%$ threshold.
    Error bars indicate one standard deviation relative to the fitting procedure.
    The gray line is the hashing bound for the associated Pauli error channel.
  }
\end{figure}

\paragraph*{Decoder.---}
In order to take account of correlations between $X$- and $Z$-type stabilizer syndromes, we implement a tensor-network approximate maximum-likelihood decoder for triangular 6.6.6 color codes following the same principles as the tensor-network decoder of Ref.~\cite{Bravyi2014} used in Ref.~\cite{Tuckett2018} for surface codes.

Consider a color code with $n$ physical qubits and $m$ independent stabilizer generators.
Let $\mathcal{P}$ denote the group of $n$-qubit Pauli operators, let $\mathcal{G}$ denote the stabilizer group, and recall that the centralizer of $\mathcal{G}$ is given by $\mathcal{C(G)} = \{ f \in \mathcal{P} : fg=gf\ \forall\ g \in \mathcal{G} \}$.
If the result of measuring the stabilizer generators is given by syndrome $s \in \{0,1\}^m$ and $f_s \in \mathcal{P}$ is some fixed Pauli operator with syndrome $s$ then the set $f_s\mathcal{C(G)}$ of all Pauli operators with syndrome $s$ is the disjoint union $f_s\mathcal{C(G)} = f_s\mathcal{G} \cup f_s\overline{X}\mathcal{G} \cup f_s\overline{Y}\mathcal{G} \cup f_s\overline{Z}\mathcal{G}$, where $\overline{X}$, $\overline{Y}$ and $\overline{Z}$ are the logical operators on the encoded qubit.

For a given syndrome $s$ and probability distribution $\pi$ on the Pauli group, the maximum-likelihood decoder can be implemented by constructing a candidate recovery operator $f_s$ consistent with $s$ and returning $\text{arg max}_f\ \pi(f\mathcal{G})$ where $f \in \{ f_s, f_s\overline{X}, f_s\overline{Y}, f_s\overline{Z} \}$ and $\pi(f\mathcal{G}) = \sum_{g \in \mathcal{G}}{\pi(fg)}$.

By analogy with the decoder of Ref.~\cite{Bravyi2014} for the surface code, we define a tensor network whose exact contraction yields the coset probability $\pi(f\mathcal{G})$ for the color code. Figures~\ref{fig:color-tensor-network}(a) and \ref{fig:color-tensor-network}(b) illustrate a distance-5 color code, whereas Figure~\ref{fig:color-tensor-network}(c) illustrates a tensor network with the same layout of qubits and stabilizers.
Bonds have dimension 4.
Stabilizer tensors are defined such that each element has a value of 1 if all indices are identical and a value of 0 otherwise.
Qubit tensors are defined such that each element has the single-qubit probability $\pi$ of the product of the restriction of $f$ to that qubit with the Paulis associated with bond indices where indices map to Paulis as $0 \mapsto I$, $1 \mapsto X$, $2 \mapsto Y$, and $3 \mapsto Z$.
In this way, all possible combinations of stabilizers are applied to $f$, and the exact contraction of such a tensor network yields the coset probability $\pi(f\mathcal{G})$.

\begin{figure}[ht]
  \includegraphics{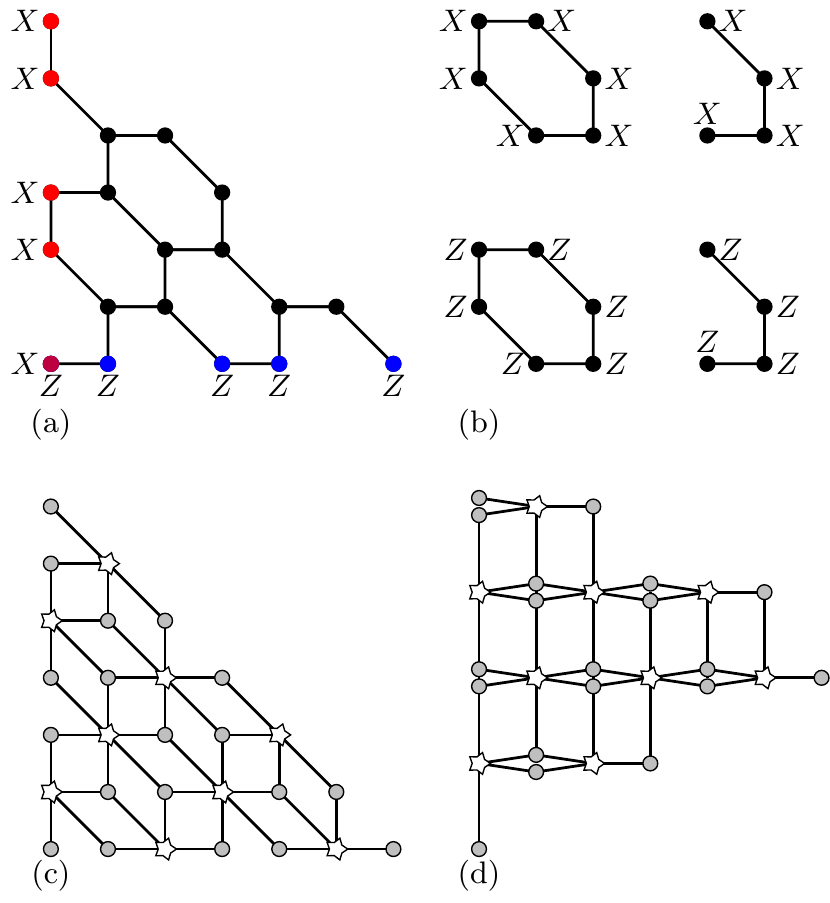}
  \caption{\label{fig:color-tensor-network}
    (a) Distance-5 triangular 6.6.6 color code with logical operators given by a product of $Z$ along the bottom edge and a product of $X$ along the left edge.
    (b) Color-code stabilizers.
    (c) Tensor network corresponding to the coset probability of a distance-5 color code; gray disks represent qubit tensors; white stars represent stabilizer tensors; and lines represent bonds.
    (d) Equivalent tensor network as a square lattice.
  }
\end{figure}

The exact contraction of the tensor network is inefficient with a runtime exponential in the number of qubits $n$.
However, by merging neighboring qubit tensors in pairs, the tensor network can be transformed into a square lattice [see Fig.~\ref{fig:color-tensor-network}(d)] so that techniques, used in the decoder of Ref.~\cite{Bravyi2014}, can be applied to efficiently approximate the coset probability.
The approximation is controlled by a parameter $\chi$ which defines the maximum bond dimension retained as the tensor network is contracted.
We refer the reader to Ref.~\cite{Bravyi2014} for full details of the approximate contraction algorithm.
We find that the performance of the decoder converges well for $\chi=36$ across all noise biases, see below.

\paragraph*{Numerics.---}
We follow the general approach taken in Ref.~\cite{Tuckett2018}; we give a brief summary here and refer the reader to Ref.~\cite{Tuckett2018} for full details.
We use triangular 6.6.6 color codes of distances $d=7, 11, 15,$ and $19$.
We estimate the threshold for biases $\eta=0.5, 1, 3, 10, 30, 100, 300, 1000, \infty$, where $\eta=p_Y/(p_X+p_Z)$ and $p_X=p_Z$, such that $\eta=0.5$ corresponds to standard depolarizing noise and $\eta=\infty$ corresponds to pure $Y$ noise (see Sec.~\ref{sec:definitions}).
We approximate maximum-likelihood decoding using the decoder, described above, with approximation parameter $\chi=36$.
The decoder converges well (generally better than in Ref.~\cite{Tuckett2018}) across the full range of biases with the weakest convergence in the low-bias regime, see Fig.~\ref{fig:color-convergence}.
We run 30\,000 simulations per code distance and physical error probability.
As in Ref.~\cite{Tuckett2018}, we use the critical exponent method of Ref.~\cite{Wang2003} to obtain threshold estimates with jackknife resampling over the code distances to determine error bounds.

\begin{figure}[ht]
  \includegraphics{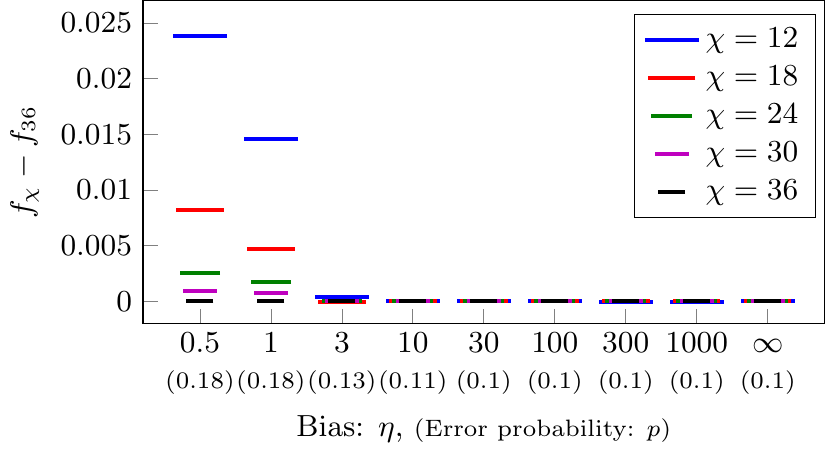}
  \caption{\label{fig:color-convergence}
    Decoder convergence for the distance $d=19$ triangular 6.6.6 color code, represented by shifted logical failure rate $f_\chi-f_{36}$, as a function of $\chi$ at a physical error probability $p$ near the threshold for the given bias $\eta$ .
    Each data point corresponds to 60\,000 runs with identical errors generated across all $\chi$ for a given bias.
  }
\end{figure}

\section{\label{sec:y-decoder}Exact optimal Y-decoder}

Here, we define the exact optimal decoder for pure $Y$ noise that we use in our numerical simulations of Sec.~\ref{sec:co-prime-advantage-pure}.
As mentioned in Sec.~\ref{sec:code-equivalence}, it is possible to decode $Y$ noise on the planar code by treating it as the concatenation of a cycle code and repetition codes and decoding level by level.
However, while efficient, such a decoder is not necessarily optimal.
Also, as mentioned in Sec.~\ref{sec:y-threshold}, the performance of the approximate maximum-likelihood decoder~\cite{Bravyi2014} used in previous studies~\cite{Tuckett2018} is found to saturate with pure $Y$ noise when tuned for efficiency.
Here, we explicitly define an exact maximum-likelihood decoder for the surface code with pure $Y$ noise that is efficient for $j{\times}k$ surface code families with small $\gcd(j,k)$, such as coprime codes, and tractable for moderate-sized square codes.

Consider a surface code with $n$ physical qubits and $m$ independent vertex and plaquette stabilizer generators.
In the case of pure $Y$ noise, the only possible error configurations are $Y$-type Pauli operators, i.e.\ operators consisting only of $Y$ and identity single-qubit Paulis.
Let $\mathcal{P}_Y$ denote the group of $n$-qubit $Y$-type Pauli operators, let $\mathcal{G}_Y$ denote the group of $Y$-type stabilizers, and define the centralizer of $\mathcal{G}_Y$ as $\mathcal{C(G_\mathnormal{Y})} = \{ f \in \mathcal{P}_Y : fg=gf\ \forall\ g \in \mathcal{G}_Y \}$.
If the result of measuring the vertex and plaquette stabilizer generators is given by syndrome $s \in \{0,1\}^m$ and $f_s \in \mathcal{P}_Y$ is some fixed $Y$-type Pauli operator with syndrome $s$ then the set $f_s\mathcal{C(G_\mathnormal{Y})}$ of all $Y$-type Pauli operators with syndrome $s$ is the disjoint union $f_s\mathcal{C(G_\mathnormal{Y})} = f_s\mathcal{G}_Y \cup f_s\overline{L}\mathcal{G}_Y$, where $\overline{L}$ is one of the single class of logical operators possible with pure $Y$ noise.

For a given syndrome $s$ and probability distribution $\pi$ on the Pauli group, the maximum-likelihood decoder for pure $Y$ noise can be implemented by constructing a candidate $Y$-type recovery operator $f_s$ consistent with $s$ and returning $\text{arg max}_f\ \pi(f\mathcal{G}_Y)$ where $f \in \{ f_s, f_s\overline{L} \}$ and $\pi(f\mathcal{G}_Y) = \sum_{g \in \mathcal{G}_Y}{\pi(fg)}$.

On a $j{\times}k$ surface code, the size of the group of $Y$-type stabilizers is $\lvert \mathcal{G}_Y \rvert = c_Y = 2^{g-1}$, where $g = \gcd(j, k)$; see Corollary~\ref{cor:count}.
Therefore, for surface codes with small $g$, such as coprime codes, the $Y$-decoder is efficient, provided that a candidate $Y$-type recovery operator $f_s$, the group of $Y$-type stabilizers $\mathcal{G}_Y$, and logical operator $\overline{L}$ can be constructed efficiently.
In the next two subsections, we describe these constructions.

\subsection{Constructing Y-type stabilizers and logical operators}

The construction of $Y$-type stabilizers and logical operators for a $j{\times}k$ code is illustrated in Fig.~\ref{fig:y-logicals}.
A minimum-weight $Y$-type logical operator is constructed by applying $Y$ operators along a path starting at the top-left corner of the lattice and descending diagonally to the right, reflecting at boundaries, until another corner is encountered from within the lattice.
We construct $Y$-type stabilizers similarly, starting at each of the next $\gcd(j,k)-1$ qubits of the top row and reflecting until the path cycles.
Together, these stabilizers generate the full group of $2^{g-1}$ $Y$-type stabilizers, and combine with the minimum-weight logical operator to give the $2^{g-1}$ $Y$-type logical operators of the $j{\times}k$ code.

\begin{figure}[ht]
  \includegraphics{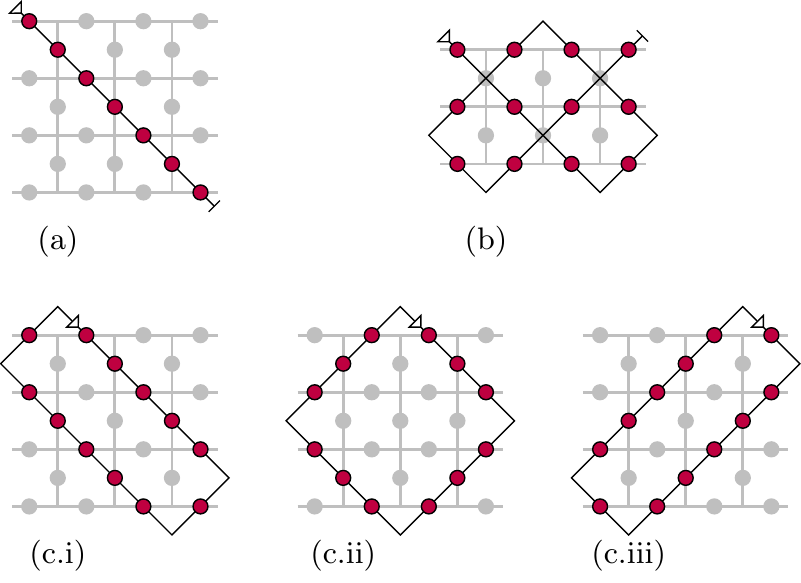}
  \caption{\label{fig:y-logicals}
    Examples of $Y$-type stabilizer and logical operator construction by applying $Y$ operators along the indicated path until a corner is encountered or the path cycles.
    Minimum-weight $Y$-type logical operators (a) and (b) for square $4{\times}4$ and coprime $3{\times}4$ codes, respectively, are constructed by starting at the top-left qubit.
    Generators of the group of $Y$-type stabilizers (c) for the square $4{\times}4$ code are constructed by starting at each of the next $\gcd(j,k)-1=3$ qubits of the top row.
    (For coprime codes, there are no $Y$-type stabilizers other than the identity.)
  }
\end{figure}

\subsection{Constructing candidate Y-type recovery operators}
The construction of a candidate $Y$-type recovery operator, consistent with a given syndrome, depends on whether the code is coprime, square, or neither.

\begin{figure}[ht]
  \includegraphics{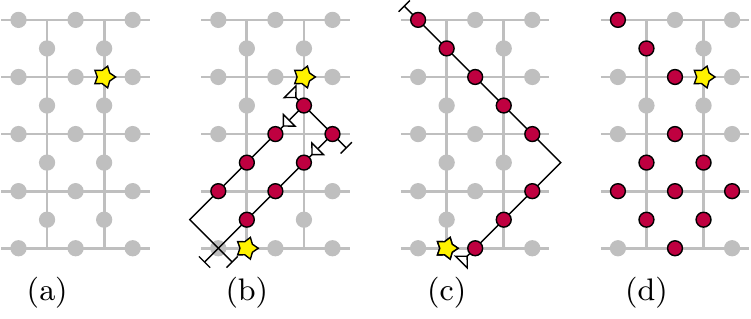}
  \caption{\label{fig:co-prime-destabilizer}
    Example of $Y$-type destabilizer construction for a coprime code.
    (a) Single syndrome location.
    (b) A partial recovery operator is constructed by applying $Y$ operators, from below the syndrome location along a diagonal to any boundary, then from that diagonal along perpendicular diagonals, until the bottom boundary is encountered.
    (c) Residual recovery operators are constructed by applying $Y$ operators, from right of each residual boundary syndrome location along a diagonal away, until a corner is encountered.
    (d) The destabilizer is a product of partial and residual recovery operators.
  }
\end{figure}

\begin{figure}[ht]
  \includegraphics{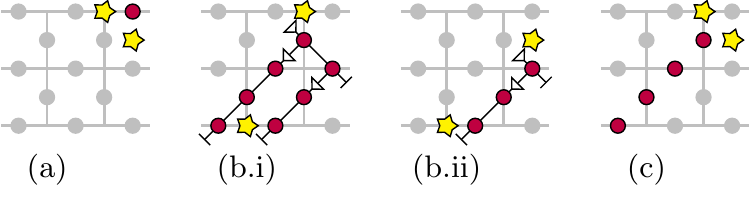}
  \caption{\label{fig:square-recovery}
    Example of candidate $Y$-type recovery operator construction for a square code using partial recovery operators.
    (a) Original error and complete syndrome.
    (b) Partial recovery operators with residual boundary syndrome locations.
    (c) The candidate recovery operator is the product of all partial recovery operators, since residual boundary syndrome locations cancel in the case of square codes.
  }
\end{figure}

For coprime codes, it is possible to construct an operator, consisting only of $Y$ and identity single-qubit Paulis, that anticommutes with any single syndrome location.
We refer to such operators as $Y$-type destabilizers.
Given a complete syndrome, a candidate $Y$-type recovery operator is then simply constructed by taking the product of $Y$-type destabilizers for each syndrome location.
One way to construct $Y$-type destabilizers for coprime codes is illustrated in Fig.~\ref{fig:co-prime-destabilizer}.
For a given syndrome location, a partial recovery operator is constructed by applying seed $Y$ operators along a path starting directly below the syndrome location and descending diagonally to the right until a boundary is encountered; further $Y$ operators are applied along paths descending diagonally to the left of each of these seed $Y$ operators, reflecting at boundaries, until the bottom boundary is encountered.
The partial recovery operator then anticommutes with the original syndrome location and residual syndrome locations on the bottom boundary.
A residual recovery operator is constructed for each residual syndrome location by applying $Y$ operators along a line starting directly to the right of the syndrome location and ascending diagonally to the right, reflecting at boundaries, until a corner is encountered from within the lattice.
The residual recovery operators then anticommute with the residual syndrome locations.
The destabilizer for the original syndrome location is then simply the product of the partial and residual recovery operators.

For square codes, $Y$-type destabilizers do not exist, in general, and, hence, a different approach to constructing a candidate $Y$-type recovery operator must be adopted.
Given a complete syndrome for a square code, a candidate $Y$-type recovery operator can be constructed by taking the product of partial recovery operators for each syndrome location, since the residual boundary syndrome locations cancel in the case of square codes; see Fig.~\ref{fig:square-recovery}.

For surface codes that are neither coprime nor square, a candidate $Y$-type recovery operator is constructed by dividing the lattice into a coprime region and square regions.
Partial recovery operators are constructed for each region leaving residual syndrome locations only on plaquettes between regions.
Residual syndrome locations can then be moved off the lattice using $Y$-type stabilizers on the square regions.

\end{document}